\documentclass[letterpaper,11pt,onecolumn,draftcls]{IEEEtran}  % Comment this line out

\usepackage{epsfig} % for postscript graphics files
\usepackage{amsmath} % assumes amsmath package installed
\usepackage{amssymb}  % assumes amsmath package installed

\usepackage{cite}      % Written by Donald Arseneau
\usepackage{subfigure} % Written by Steven Douglas Cochran
\usepackage{url}       % Written by Donald Arseneau
\usepackage{verbatim}   % multi-line comment: \begin{comment},\end{comment}
\usepackage{tabularx}

\usepackage{algorithm}
\usepackage{algpseudocode}

\usepackage{bbold}
\usepackage[dvips]{color}

\newcommand{\Rate}{\mathcal{R}}

\begin{document}

\title{\LARGE \bf
Low-Complexity Scheduling Policies for Achieving Throughput and 
Asymptotic Delay Optimality in Multi-Channel Wireless Networks} 
%Low-complexity Delay-efficient Scheduling for OFDM Downlink}

%\author{ \parbox{3 in}{\centering Huibert Kwakernaak*
%         \thanks{*Use the $\backslash$thanks command to put information here}\\
%         Faculty of Electrical Engineering, Mathematics and Computer Science\\
%         University of Twente\\
%         7500 AE Enschede, The Netherlands\\
%         {\tt\small h.kwakernaak@autsubmit.com}}
%         \hspace*{ 0.5 in}
%         \parbox{3 in}{ \centering Pradeep Misra**
%         \thanks{**The footnote marks may be inserted manually}\\
%        Department of Electrical Engineering \\
%         Wright State University\\
%         Dayton, OH 45435, USA\\
%         {\tt\small pmisra@cs.wright.edu}}
%}
%\author{Bo Ji, Changhee Joo and Ness B. Shroff}
%\author{}

%\author{Bo~Ji, Gagan~R.~Gupta, Xiaojun~Lin, and Ness~B.~Shroff}
\author{Bo~Ji, Gagan~R.~Gupta, Xiaojun~Lin, and Ness~B.~Shroff
\thanks{B. Ji is with AT\&T Labs. 
%G. R. Gupta is with AT\&T Labs.
X. Lin is with School of ECE at Purdue University.
N. B. Shroff is with Departments of ECE 
and CSE at the Ohio State University.
Emails: ji.33@osu.edu, gagan.gupta@iitdalumni.com, linx@ecn.purdue.edu, shroff.11@osu.edu.}
%\thanks{Emails: ji@ece.osu.edu, gagan.gupta@iitdalumni.com, linx@ecn.purdue.edu, shroff@ece.osu.edu.}
}

\newcounter{theorem}
\newtheorem{theorem}{Theorem}
\newtheorem{proposition}[theorem]{\it Proposition}
\newtheorem{lemma}[theorem]{\it Lemma}
\newtheorem{corollary}[theorem]{\it Corollary}
\newcounter{definition}
\newtheorem{definition}{\it Definition}
\newcounter{assumption}
\newtheorem{assumption}{\it Assumption}

\newcommand{\Graph}{\mathcal{G}}
\newcommand{\Node}{\mathcal{V}}
\newcommand{\Vertex}{\mathcal{V}}
\newcommand{\Edge}{\mathcal{E}}
\newcommand{\Link}{\mathcal{L}}
\newcommand{\Flow}{\mathcal{S}}
\newcommand{\Route}{\mathcal{H}}
\newcommand{\A}{\mathcal{A}}
\newcommand{\B}{\mathcal{B}}
\newcommand{\C}{\mathcal{C}}
\newcommand{\bZ}{\mathbb{Z}}
\newcommand{\uZ}{\underline{Z}}
\newcommand{\uP}{\underline{P}}
\newcommand{\uR}{\underline{R}}
\newcommand{\uB}{\underline{B}}
\newcommand{\uU}{\underline{U}}
\newcommand{\uQ}{\underline{Q}}
\newcommand{\utP}{\tilde{\underline{P}}}
\newcommand{\utB}{\tilde{\underline{B}}}
\newcommand{\tB}{\tilde{B}}
\newcommand{\Hop}{{H}}
\newcommand{\Matching}{\mathcal{M}}
\newcommand{\Pair}{\mathcal{P}}
\newcommand{\Int}{\mathbb{Z}}
\newcommand{\Expect}{\mathbf{E}}
\newcommand{\System}{\mathcal{Y}}
\newcommand{\Markov}{\mathcal{X}}
\newcommand{\Prob}{\mathbb{P}}
\newcommand{\vM}{\vec{M}}
\newcommand{\bM}{\mathbf{M}}
\newcommand{\s}{{(s)}}
\newcommand{\sk}{{s,k}}
\newcommand{\hsk}{{\hat{s},\hat{k}}}
\newcommand{\rj}{{r,j}}
\newcommand{\xm}{{x_m}}
\newcommand{\xml}{{x_{m_l}}}
\newcommand{\UGraph}{{U}}
\newcommand{\UVertex}{{X}}
\newcommand{\UEdge}{{Y}}
\newcommand{\vlambda}{\vec{\lambda}}
\newcommand{\vpi}{\vec{\pi}}
\newcommand{\vphi}{\vec{\phi}}
\newcommand{\vpsi}{\vec{\psi}}
\newcommand{\blambda}{\mathbf{\lambda}}
\newcommand{\tM}{\tilde{M}}
\newcommand{\tpi}{\tilde{\pi}}
\newcommand{\tphi}{\tilde{\phi}}
\newcommand{\tpsi}{\tilde{\psi}}
\newcommand{\vmu}{\vec{\mu}}
\newcommand{\vnu}{\vec{\nu}}
\newcommand{\valpha}{\vec{\alpha}}
\newcommand{\vbeta}{\vec{\beta}}
\newcommand{\ve}{\vec{e}}
\newcommand{\vxi}{\vec{\xi}}
\newcommand{\hgamma}{\gamma}
\newcommand{\GMSLambda}{\Lambda_{\text{\it GMS}}}
\newcommand{\argmax}{\operatornamewithlimits{argmax}}
\newcommand{\card}{\aleph}
\newcommand{\V}{\mathcal{V}}
\newcommand{\X}{\mathcal{X}}
\newcommand{\Y}{\mathcal{Y}}
\newcommand{\Z}{\mathcal{Z}}
\newcommand{\mS}{\mathcal{S}}
\newcommand{\mM}{\mathcal{M}}
\newcommand{\mE}{\mathcal{E}}
\newcommand{\indicator}{\mathbb{1}}
\newcommand{\mA}{\mathcal{A}}

\maketitle
% \thispagestyle{empty}
% \pagestyle{empty}

%%%%%%%%%%%%%%%%%%%%%%%%%%%%%%%%%%%%%%%%
\begin{abstract}
%%%%%%%%%%%%%%%%%%%%%%%%%%%%%%%%%%%%%%%%
In this paper, we study the scheduling problem for downlink 
transmission in a multi-channel (e.g., OFDM-based) wireless
network. We focus on a single cell, with the aim of developing 
a unifying framework for designing low-complexity scheduling 
policies that can provide optimal performance in terms of both 
throughput and delay. We develop new easy-to-verify sufficient 
conditions for rate-function delay optimality (in the many-channel 
many-user \emph{asymptotic} regime) and throughput optimality
(in general \emph{non-asymptotic} setting), respectively. 
The sufficient conditions allow us to prove rate-function delay
optimality for a class of Oldest Packets First (OPF) policies and 
throughput optimality for a large class of Maximum Weight in the 
Fluid limit (MWF) policies, respectively. By exploiting the special 
features of our carefully chosen sufficient conditions and intelligently 
combining policies from the classes of OPF and MWF policies, we 
design hybrid policies that are both rate-function delay-optimal 
and throughput-optimal with a complexity of $O(n^{2.5} \log n)$, 
where $n$ is the number of channels or users. Our sufficient condition 
is also used to show that a previously proposed policy called Delay 
Weighted Matching (DWM) is rate-function delay-optimal. However, DWM 
incurs a high complexity of $O(n^5)$. Thus, our approach yields 
significantly lower complexity than the only previously designed delay 
and throughput optimal scheduling policy. We also conduct numerical 
experiments to validate our theoretical results.
\end{abstract}

% \begin{IEEEkeywords}
% Scheduling, Multi-channel, Wireless networks, OFDM, Throughput optimality, 
% Delay optimality, Low-complexity, Quality of Service, Large-Deviations Theory
% \end{IEEEkeywords}

%%%%%%%%%%%%%%%%%%%%%%%%%%%%%%%%%%%%%%%%
\section{Introduction} \label{sec:intro}
%%%%%%%%%%%%%%%%%%%%%%%%%%%%%%%%%%%%%%%%
Designing high-performance scheduling algorithms has been a vital and 
challenging problem in wireless networks. Among the many dimensions 
of network performance, the most critical ones are perhaps throughput, 
delay, and complexity. However, it is in general extremely difficult, 
if not impossible, to develop scheduling policies that attain the optimal 
performance in terms of both throughput and delay, without the cost of 
high complexity \cite{shah11}.

In this paper, we focus on the setting of a single-hop multi-user multi-channel 
system. A practically important example of such a multi-channel system is the 
downlink of a single cell in 4G OFDM-based celluar networks (e.g., LTE and 
WiMax). Such a system typically has a large bandwidth that can be divided into 
multiple orthogonal sub-bands (or channels), which need to be allocated to a 
large number of users by a scheduling algorithm. The main question that we will 
attempt to answer in this paper is the following: \emph{How do we design efficient 
scheduling algorithms that simultaneously provide high throughput, small delay, 
and low complexity?} 

We consider a multi-channel system that has $n$ channels and a proportionally 
large number of users. This setting is referred to as the many-channel many-user 
asymptotic regime when $n$ goes infinity. The connectivity between each user and 
each channel is assumed to be time-varying, due to channel fading. We assume that 
the base station (BS) maintains separate First-in First-out (FIFO) queues that 
buffer the packets destined to each user. The \emph{delay} metric that we will 
focus on in this paper is the \emph{asymptotic decay-rate} (also called the 
\emph{rate-function} in the large-deviations theory) of the probability that the 
largest packet waiting time in the system exceeds a fixed threshold, as both the 
number of channels and the number of users go to infinity. (Refer to Eq.~(\ref{eq:rf}) 
for the precise definition.) 

Next, we overview some key related works. In \cite{tassiulas93}, the authors
considered a single-server model with time-varying channels, and showed that 
the longest-connected-queue (LCQ) algorithm minimizes the average delay for 
the special case of symmetric (\emph{i.i.d.} Bernoulli) arrival and channel. 
Later, the results were generalized for a multi-server model in \cite{ganti07}. 
The authors of \cite{kittipiyakul09} further generalized the multi-server 
model by considering more general permutation-invariant arrivals (that are not 
restricted to Bernoulli only) and multi-rate channel model. Hence, the problem 
of minimizing a general cost function of queue-lengths (includes minimizing the 
expected delay) studied in \cite{kittipiyakul09} becomes harder. There, for 
special cases of ON-OFF channel model with two users or allowing for fractional 
server allocation, an optimal scheduling algorithm was derived. Using the insights 
obtained from the analytical results in \cite{kittipiyakul09} for ON-OFF channel 
model, in \cite{kittipiyakul07} the same authors developed heuristic policies and 
showed through simulations that their proposed heuristic policies perform well 
under a general channel model. Note that in contrast to this paper, the above 
studies directly minimize queue-length or delay in a \emph{non-asymptotic} regime, 
which is an extremely difficult problem in general. 

As we do in this paper, another body of related works \cite{bodas09,bodas10,
bodas11a,bodas11b} focus on the many-channel many-user asymptotic regime, 
where the analysis may become more tractable. Even though the analysis for 
an asymptotic setting is very different from the non-asymptotic analysis in 
\cite{kittipiyakul09}, it is remarkable that some of the insights are consistent. 
For example, from a delay optimality perspective, the above two bodies of 
studies both point to the tradeoff between maximizing instantaneous throughput 
and balancing the queues. Thus, we believe that, collectively, these studies 
under different settings provide useful insights for designing efficient 
scheduling solutions in practice. 
 
In \cite{bodas09,bodas10,bodas11a,bodas11b}, a number of queue-length-based 
scheduling policies for achieving optimal or positive queue-length-based 
rate-unction\footnote{The queue-length-based rate-function is defined as the 
asymptotic decay-rate of the probability that the largest queue length in the 
system exceeds a fixed threshold.} were developed. In particular, an optimal 
scheduling policy that maximizes the queue-length-based rate-function has been 
derived with complexity $O(n^3)$ \cite{bodas11b}. However, these works have two 
key limitations. First, the schedulers' performance are proven under the 
assumption that the arrival process is \emph{i.i.d. not only across users, but 
also in time}, which does not model the temporal correlation present in most 
real network traffic. More importantly, it is well known that good queue-length 
performance does not necessarily translate to good delay performance 
\cite{sharma11b,sharma11,ji13c}. A recently developed scheduling policy called Delay 
Weighted Matching (DWM) \cite{sharma11,sharma11b}, which makes scheduling decisions 
by maximizing the sum of the delays of the scheduled packets in each time-slot, 
focuses directly on the delay performance as we do in this paper. It has been shown 
that the DWM policy is rate-function delay-optimal in some cases. However, DWM has 
the following two key drawbacks: 1) it is unclear whether DWM is rate-function 
delay-optimal in general; and 2) DWM yields a very high complexity of $O(n^5)$ 
and is thus not amenable for practical implementations.

Hence, the state-of-the-art does not satisfactorily answer our main
question of how to design scheduling policies with a low complexity,
while guaranteeing \emph{provable optimality} for both throughput 
and delay. In this paper, we address this challenge, and provide the 
following key intellectual contributions.
%
%To address this challenge, \emph{our overall 
%goal in this paper is to develop a unifying framework for designing 
%low-complexity scheduling policies that can achieve the optimal 
%performance in terms of both throughput and delay.} 

First, we characterize \emph{easy-to-verify} sufficient conditions for 
rate-function delay optimality in the many-channel many-user asymptotic 
regime and for throughput optimality in general non-asymptotic settings. 
The sufficient conditions allow us to prove rate-function delay optimality 
for a class of \emph{Oldest Packets First (OPF)} policies and throughput 
optimality for a large class of \emph{Maximum Weight in the Fluid limit 
(MWF)} policies. Moreover, the sufficient conditions can be used to show 
that a slightly modified version of the DWM policy is both rate-function 
delay-optimal and throughput-optimal.

Second, we develop an $O(n^{2.5} \log n)$-complexity scheduling policy
called DWM-$n$. The DWM-$n$ policy shares the high-level similarity with 
the DWM policy \cite{sharma11,sharma11b}, but makes scheduling decisions 
in each time-slot by maximizing the sum of the delays of the scheduled 
packets over only the $n$ oldest packets in the system, rather than over 
all the packets as in the DWM policy. We show that DWM-$n$ is an OPF policy 
and is thus rate-function delay-optimal. However, DWM-$n$ is \emph{not} 
throughput-optimal in general, and may perform poorly when $n$ is not large.

Third, by exploiting the special features of our carefully-chosen sufficient 
conditions and intelligently combining policies from the classes of OPF and 
MWF policies, we develop a class of two-stage hybrid policies that are both 
throughput-optimal and rate-function delay-optimal. 
% The basic idea is as 
% follows. In stage~1, we choose an OPF policy and focus on scheduling the 
% $n$ oldest packets. This not only guarantees rate-function delay optimality, 
% \emph{but also satisfies the sufficient condition for throughput optimality for 
% all the selected servers in stage~1}. In stage~2, the selected servers in stage~1 
% will not be considered. For the remaining servers, we pick a policy from the class 
% of MWF policies. Since the chosen MWF policy is run over the remaining servers 
% that were not selected in stage~1, it ensures that the sufficient condition for 
% throughput optimality is satisfied for these remaining servers. Further, since 
% the packets and servers matched in stage~1 are not touched, the satisfaction of 
% the sufficient condition for rate-function delay optimality is not perturbed. 
In particular, we can adopt the DWM-$n$ policy in stage~1 and the Delay-based MaxWeight 
Scheduling (D-MWS) policy in stage~2, respectively, so as to \emph{design an optimal 
hybrid policy with a low complexity of $O(n^{2.5} \log n)$}.  

Finally, we conduct numerical experiments to validate our theoretical 
results in different scenarios.

%%%%%%%%%%%%%%%%%%%%%%%%%%%%%%%%%%%%%%%%%%%
\section{System Model} \label{sec:model}
%%%%%%%%%%%%%%%%%%%%%%%%%%%%%%%%%%%%%%%%%%%
We consider a multi-channel system with $n$ orthogonal channels and $n$ 
users, which can be modeled as a multi-queue multi-server system with 
stochastic connectivity, as shown in Fig.~\ref{fig:system}. \emph{For
ease of presentation, the number of users is assumed to be equal to 
the number of channels. Our analysis for rate-function delay optimality 
follows similarly if the number of users scales linearly with the number 
of channels.} Throughout the rest of the paper, we will use the terms 
``user" and ``queue" interchangeably, and use the terms ``channel" and 
``server" interchangeably. We assume that time is slotted. In a time-slot, 
a server can be allocated to only one queue, but a queue can get service 
from multiple servers. The connectivity between queues and servers is 
time-varying, i.e., it can change between ``ON" and ``OFF" from time to 
time. We assume that perfect channel state information (i.e., whether each 
channel is ON or OFF for each user in each time-slot) is known at the BS. 
This is a reasonable assumption in the downlink scenario of a single cell 
in a multi-channel cellular system with dedicated feedback channels.

\begin{figure}[t]
\centering
\epsfig{file=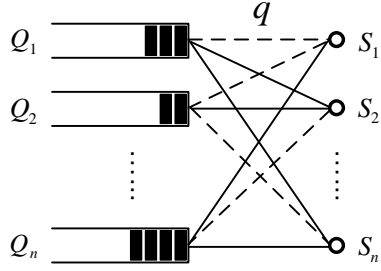,width=0.3\linewidth}
\caption{System model. The connectivity between each pair of queue $Q_i$ and 
server $S_j$ is ``ON" (denoted by a solid line) with probability $q$, and ``OFF" 
(denoted by a dashed line) otherwise.}
\label{fig:system}
\end{figure}

The notations used in this paper are as follows. We let $Q_i$ denote 
the FIFO queue (at the BS) associated with the $i$-th user, and let 
$S_j$ denote the $j$-th server. We assume infinite buffer for all the
queues. Let $A_i(t)$ denote the number of packet arrivals to queue 
$Q_i$ in time-slot $t$, let $A(t)=\sum_{i=1}^n A_i(t)$ denote the 
cumulative arrivals to the entire system in time-slot $t$, and let 
$A(t_1,t_2) =\sum_{\tau=t_1}^{t_2} A(\tau)$ denote the cumulative 
arrivals to the system from time $t_1$ to $t_2$. We let $\lambda_i$ 
be the mean arrival rate of queue $Q_i$, and let $\lambda \triangleq 
[\lambda_1,\lambda_2,\dots,\lambda_n]$ denote the arrival rate vector.
We assume that packet 
arrivals occur at the beginning of each time-slot, and packet departures 
occur at the end of each time-slot. We let $Q_i(t)$ denote the length 
of queue $Q_i$ at the beginning of time-slot $t$ immediately after packet 
arrivals. Also, let $Z_{i,l}(t)$ denote the delay (i.e., waiting time) 
of the $l$-th packet of queue $Q_i$ at the beginning of time-slot $t$, 
\emph{which is measured since the time when the packet arrived to queue 
$Q_i$ until the beginning of time-slot $t$}. Note that at the end of each 
time-slot, the packets still present in the system will have their delays 
increased by one due to the elapsed time. We then let $W_i(t)=Z_{i,1}(t)$ 
denote the head-of-line (HOL) packet delay of queue $Q_i$ at the beginning 
of time-slot $t$. Further, we use $C_{i,j}(t)$ to denote the capacity of 
the link between queue $Q_i$ and server $S_j$ in time-slot $t$, i.e., the 
maximum number of packets that can be served by server $S_j$ from queue 
$Q_i$ in time-slot $t$. 
% Let $Y_{i,j}(t)$ be the service of queue $Q_i$ received from server 
% $S_j$ in time-slot $t$, i.e., $Y_{i,j}(t)=C_{i,j}(t)$ if server $S_j$ is 
% allocated to serve queue $Q_i$, and $Y_{i,j}(t)=0$ otherwise. 
% Finally, we define $D(x||y) \triangleq x \log \frac {x} {y} + (1-x) \log \frac {1-x} 
% {1-y}$ and $(x)^+ \triangleq \max(x,0)$, 
Finally, we let $\mathbb{1}_{\{ \cdot \}}$ denote the indicator function, 
and let $\mathbb{Z}^+$ denote the set of positive integers.

% Since we use the delay as the weight of a packet or a queue, in the rest 
% of this paper, we may use the terms ``delay" and ``weight" interchangeably, 
% whenever there is no confusion.

We now state the assumptions on the arrival processes. The throughput analysis 
is carried out under very general conditions (Assumption~\ref{ass:arr_slln}) 
similar to that of \cite{andrews04}. 

\begin{assumption}
\label{ass:arr_slln}
For each user $i \in \{1,2,\dots,n \}$, the arrival process $A_i(t)$ 
is an irreducible and positive recurrent Markov chain with countable 
state space, and satisfies the Strong Law of Large Numbers: That is, 
with probability one,
\begin{equation}
\label{eq:slln}
\lim_{t \rightarrow \infty} \frac {\sum^{t-1}_{\tau=0} A_i(\tau)} {t} = \lambda_i.
\end{equation}
We also assume that the arrival processes are mutually independent
across users (which can be relaxed for showing throughput optimality, 
as discussed in \cite{andrews04}.)
\end{assumption}

Assumptions~\ref{ass:arr_bound} and \ref{ass:arr_ld} (stated below) 
will be used for rate-function delay analysis.

\begin{assumption}
\label{ass:arr_bound}
There exists a finite $L$ such that $A_i(t) \le L$ for any $i$ and $t$, 
i.e., arrivals are bounded. Further, we assume $\Prob (A(s,s+t-1) = Lnt) 
> 0$ for any $s$, $t$ and $n$.
\end{assumption}

\begin{assumption}
\label{ass:arr_ld}
The arrival processes are \emph{i.i.d.} across users, and $\lambda_i=p$ for any 
user $i$. Given any $\epsilon>0$ and $\delta>0$, there exists $T_B(\epsilon,\delta)>0$, 
$N_B(\epsilon,\delta)>0$, and a positive function $I_B(\epsilon,\delta)$ independent 
of $n$ and $t$ such that
\[
\Prob \left( \frac {\sum_{\tau=1}^{t} \mathbb{1}_{\{|A(\tau) - pn|
> \epsilon n \}}} {t} > \delta \right) < \exp (-nt I_B(\epsilon,\delta)),
\]
for all $t \ge T_B(\epsilon,\delta)$ and $n \ge N_B(\epsilon,\delta)$. 
\end{assumption}

Assumptions~\ref{ass:arr_bound} and \ref{ass:arr_ld} are relatively mild. 
The first part of Assumption~\ref{ass:arr_bound} and Assumption~\ref{ass:arr_ld}
have also been used in the previous work \cite{sharma11,sharma11b} for 
rate-function delay analysis. 
In Assumption~\ref{ass:arr_bound}, the first part requires that the arrivals 
in each time-slot have bounded support; and the second part guarantees that 
there is a positive probability that all users have the maximum number of
arrivals in any time-interval with any length. 
Assumption~\ref{ass:arr_ld} allows the arrivals for each user to be 
correlated over time (e.g., arrivals driven by a two-state Markov chain), 
which is more general than the arrival processes (\emph{i.i.d.} in time) 
considered in \cite{bodas09,bodas10,bodas11a,bodas11b}. 

We then describe our channel model as follows.

\begin{assumption}
\label{ass:channel}
In any time-slot $t$, $C_{i,j}(t)$ is modeled as a Bernoulli 
random variable with a parameter $q \in (0,1)$, i.e.,
\[
C_{i,j}(t) = \left\{
\begin{array}{ll}
1, & \text{with probability}~ q,\\
0, & \text{with probability}~ 1-q.
\end{array}
\right. 
\]
All the random variables $C_{i,j}(t)$ are assumed to be mutually independent 
across all the variables $i,j$ and $t$. 
\end{assumption}

We assume unit channel capacity as above. Under this assumption, 
we will also let $C_{i,j}(t)$ denote the connectivity between queue $Q_i$ and 
server $S_j$ in time-slot $t$, without causing confusions. As in the previous 
works \cite{bodas09,bodas10,bodas11a,bodas11b,sharma11,sharma11b}, in this paper 
we assume \emph{i.i.d.} channels for the analytical results only. 
% The sub-bands being \emph{i.i.d.} is a reasonable assumption when the channel 
% width is larger than the coherence bandwidth of the environment. 
Moreover, we will show through simulations that our proposed low-complexity 
solution also performs well in more general scenarios, e.g., when the channel 
condition follows a two-state Markov chain that allows correlation over time. 
Further, we will briefly discuss how to generalize our solution to more general 
scenarios towards the end of this paper.

% We emphasize that our focus in this paper is to overcome the high 
% complexity issue of the DWM policy \cite{sharma11}, without 
% losing throughput optimality and rate-function delay optimality. Hence,
% as in the previous works \cite{sharma11,kittipiyakul09,bodas09,
% bodas10}, we make the \emph{i.i.d.} assumptions across users and 
% across channels for ease of analysis. However, we believe that the insights 
% obtained in this paper will be useful for more realistic scenarios (e.g., 
% possibly correlated and heterogeneous users and channels, and more general 
% channel models where the channel rate takes a finite set of values).

%%%%%%%%%%%%%%%%%%%%%%%%%%%%%%%%%%%
%\subsection{Performance Objectives}
%%%%%%%%%%%%%%%%%%%%%%%%%%%%%%%%%%%
Next, we define the \emph{optimal throughput region} (or \emph{stability 
region}) of the system for any fixed integer $n>0$. As in \cite{andrews04}, 
a stochastic queueing network is said to be \emph{stable} if it can be 
described as a discrete-time countable Markov chain and the Markov chain 
is stable in the following sense: The set of positive recurrent states is 
nonempty, and it contains a finite subset such that with probability one, 
this subset is reached within finite time from any initial state. When all 
the states communicate, stability is equivalent to the Markov chain being 
positive recurrent. The \emph{throughput region} of a scheduling policy is 
defined as the set of arrival rate vectors for which the network remains 
stable under this policy. Further, the \emph{optimal throughput region} is 
defined as the union of the throughput regions of all possible scheduling 
policies. We let $\Lambda^{*}$ denote the optimal throughput region. A 
scheduling policy is \emph{throughput-optimal}, if it can stabilize any 
arrival rate vector $\lambda$ strictly inside $\Lambda^*$. 
For more discussions on the characterization of $\Lambda^{*}$ please refer 
to Appendix~\ref{app:otr}.
% For more discussions on the characterization of $\Lambda^{*}$ please refer 
% to our online technical report \cite{ji13a}.

% Throughput is usually viewed as the first-order performance metric for 
% scheduling policies. To maximize the throughput (i.e., to achieve the 
% optimal throughput region) has been the primary focus of designing 
% efficient scheduling policies for wireless networks in many existing 
% works (e.g., \cite{tassiulas92,lin06c,georgiadis06,andrews04,eryilmaz05,
% ji13c}). However, this goal alone is not enough for applications with 
% more stringent Quality-of-Service (QOS) requirements (e.g., voice or 
% video). 
For delay analysis, we consider the many-channel many-user asymptotic regime. 
Let $W(t)$ denote the largest HOL delay over all the queues (i.e., the largest 
or worst packet waiting time in the system) at the beginning of time-slot $t$, 
i.e., $W(t) \triangleq \max_{1 \le i \le n} W_i(t)$. Assuming that the system 
is stationary and ergodic, we define the \emph{rate-function} for integer 
threshold $b \ge 0$ as
\begin{equation}
\label{eq:rf}
I(b) \triangleq \lim_{n \rightarrow \infty} \frac {-1} {n} \log \Prob (W(0)>b).
\end{equation}
We can then estimate $\Prob (W(0)>b) \approx \exp (-nI(b))$ when $n$ is 
large, and the estimation accuracy tends to be higher as $n$ increases. 
Clearly, for large $n$ a larger value of the rate-function leads to better 
delay performance, i.e., a smaller probability that the largest HOL delay 
exceeds a certain threshold. A scheduling policy is \emph{rate-function 
delay-optimal} if for any fixed integer threshold $b\ge0$, it achieves 
the maximum rate-function over all possible scheduling policies. 

Note that \emph{the rate-function optimality is studied in the asymptotic 
regime, i.e., when $n$ goes to infinity. Although the convergence of the 
rate-function is typically fast, the throughput performance may be poor 
for small to moderate values of $n$.} As a matter of fact, a rate-function 
delay-optimal policy may not even be throughput-optimal for a fixed $n$ 
(e.g., the DWM-$n$ policy that we will propose in Section~\ref{sec:suff}). 
To that end, we are interested in designing scheduling policies that maximize 
both the throughput (for any fixed $n$) and the rate-function (in the 
many-channel many-user asymptotic regime).

%%%%%%%%%%%%%%%%%%%%%%%%%%%%%%%%%%%%%%%%%%%%%%%%%%%%%%%%%%%%%
\section{An Upper Bound on The Rate-Function} \label{sec:ub}
%%%%%%%%%%%%%%%%%%%%%%%%%%%%%%%%%%%%%%%%%%%%%%%%%%%%%%%%%%%%%Define
In this section, we derive an upper bound on the rate-function that can be 
achieved by any scheduling algorithm. Then, later in Section~\ref{sec:suff},
we will provide a sufficient condition for achieving this upper bound and
thus achieving the optimal rate-function.

As in \cite{sharma11,sharma11b}, for any integer $t>0$ and any real number 
$x \ge 0$, we define the quantity
\[
I_A(t,x) \triangleq \sup_{\theta>0} [\theta(t+x)-\lambda_{A_i(-t+1,0)}(\theta)],
\]
where $\lambda_{A_i(-t+1,0)}(\theta) = \log \Expect [e^{\theta {A_i(-t+1,0)}}]$ 
is the cumulant-generating function of $A_i(-t+1,0)$ and $A_i(-t+1,0) = 
\sum_{\tau=-t+1}^{0} A_i(\tau)$. From Cramer's Theorem, this quantity, $I_A(t,x)$, 
is equal to the asymptotic decay-rate of the probability that in any interval of 
$t$ time-slots, the total number of packet arrivals to the system is no smaller 
than $n(t+x)$, as $n$ tends to infinity, i.e.,  
\begin{equation}
\label{eq:IAtx}
\lim_{n \rightarrow \infty} \frac {-1}{n} \log \Prob(A(-t+1,0) \ge n(t+x)) = I_A(t,x).
\end{equation}

Define the following for the case of $L>1$. For any integer $x \ge 0$, 
we define $t_x$ as 
\[
t_x \triangleq \frac{x} {L-1}.
\]
Then, we define $\Psi_b \triangleq \{c \in \{0,1,\dots,b \} ~|~ 
t_{b-c} \in \mathbb{Z}^+ \}$. It will later become clear why the values 
of $c$ in the set $\Psi_b$ are important and need to be considered separately.  
Let $I_X \triangleq \log \frac {1} {1-q}$. Then, for any integer $b \ge 0$, 
we define the quantity
\begin{equation}
\begin{split}
I_0(b) \triangleq \min \{ &(b+1)I_X, \\
& \min_{0 \le c \le b} \{\inf_{t>t_{b-c}} I_A(t,b-c) + c I_X \}, \\
& \min_{c \in \Psi_b} \{ I_A(t_{b-c},b-c) + (c+1) I_X \} \}.
\end{split}
\end{equation}
Further, for any given integer $L \ge 1$, we define
\[
I^*_0(b) \triangleq \left\{
\begin{array}{ll}
(b+1) I_X, &~\text{if}~ L=1,\\
I_0(b), &~\text{if}~ L>1.
\end{array}
\right. 
\]

In the following theorem, we show that for any given integer threshold 
$b \ge 0$, $I^*_0(b)$ is an upper bound of the rate-function that can 
be achieved by any scheduling policy.
\begin{theorem}
\label{thm:ub}
Given the system model described in Section~\ref{sec:model}, for any 
scheduling algorithm, we have
\[
\limsup_{n \rightarrow \infty} \frac {-1}{n} \log \Prob(W(0)>b) \le I^*_0(b),
\]
for any given integer threshold $b \ge 0$.
\end{theorem}

We prove Theorem~\ref{thm:ub} by considering three types of events that lead 
to the delay-violation event $\{W(0)>b\}$ and computing their probabilities.
We provide the proof in Appendix~\ref{app:thm:ub}.

Note that in \cite{sharma11b}, the authors derived another upper bound 
$\min \{ (b+1)I_X, \min_{0 \le c \le b} \{ I^+_A(b-c)+cI_X\} \}$, where 
$I^+_A(x) \triangleq \inf_{t>0} I^+_A(t,x)$ and $I^+_A(t,x) \triangleq 
\lim_{y \rightarrow x^+} I_A(t,y)$. We would like to remark that their
upper bound was derived by considering two types of events that lead to 
the delay-violation event, which yet accounts for only a proper subset 
of the events that we consider in Appendix~\ref{app:thm:ub}. Hence, their 
upper bound could be larger than $I^*_0(b)$ in some cases.

% In particular, for $L=1$, the two upper bounds meet at 
% $(b+1)I_X$, as the total number of packet arrivals to the system during an interval 
% of $t$ time-slots can never exceed $nt$ when $L=1$ and thus $I^+_A(t,b-c)=\infty$ 
% for all $c \in \{0,\dots,b\}$ and all $t>0$. While for $L>1$, their upper bound 
% corresponds to the minimum over the first two terms of $I_0(b)$, i.e., $\min \{ 
% (b+1)I_X, \min_{0 \le c \le b} \{ \inf_{t>t_{b-c}} I_A(t,b-c)+cI_X\} \}$, and is 
% thus no smaller than $I_0(b)$. To see this, consider any $c \in \{0,\dots,b\}$.
% From Lemma~\ref{lem:IAtx}, we have $I^+_A(t,b-c) = I_A(t,b-c)$ for all $t>t_{b-c}$, 
% since $b-c = (L-1)t_{b-c} < (L-1)t$ when $t>t_{b-c}$. Further, it is easy to see 
% that $I^+_A(t,b-c) = \infty$ for all $t \le t_{b-c}$. This is because the total 
% number of packet arrivals to the system during an interval of $t$ time-slots can 
% never exceed $Lnt$,

%%%%%%%%%%%%%%%%%%%%%%%%%%%%%%%%%%%%%%%%%%%%%%%%
\section{Sufficient Conditions} \label{sec:suff}
%%%%%%%%%%%%%%%%%%%%%%%%%%%%%%%%%%%%%%%%%%%%%%%%
In \cite{sharma11,sharma11b}, the authors proposed the DWM policy and studied
its rate-function delay optimality\footnote{Although the delay metric considered
in \cite{sharma11,sharma11b} is slightly different from ours, both metrics are 
closely related. Moreover, the rate-function delay analysis for DWM in 
\cite{sharma11,sharma11b} is also applicable for our defined rate-function as 
in (\ref{eq:rf}).} (without the second part of Assumption~\ref{ass:arr_bound}) 
in some cases. Specifically, in \cite{sharma11,sharma11b}, the authors proved 
that DWM attains a rate-function that is no smaller than $\min \{ (b+1)I_X, 
\min_{0 \le c \le b} \{ I_A(b-c)+cI_X\} \}$, where $I_A(x) \triangleq \inf_{t>0} 
I_A(t,x)$. This is proved by showing that the FBS policy (with a properly chosen 
operating parameter $h$) can attain this rate-function and DWM dominates FBS for 
all values of $h$ in a sample-path sense. As pointed out in \cite[Section~V.D]{sharma11b}, 
there may be a gap between the rate-function attained by DWM and the upper bound 
derived in \cite{sharma11b}, depending on the value of $b$ and the arrival process. 
More specifically, it can be shown that for given $b\ge0$, if $I_A(b-c)=I^+_A(b-c)$ 
for all values of $c \in \{0,\dots,b\}$ for the given arrival process, then both FBS 
and DWM are rate-function delay-optimal.

However, it is unclear whether the DWM policy is rate-function delay-optimal in 
general. Moreover, its high complexity $O(n^5)$ renders it impractical. Hence, 
\emph{the grand challenge is to find low-complexity scheduling policies that are 
both throughput-optimal and rate-function delay-optimal.} To that end, in this 
section, we first characterize easy-to-verify sufficient conditions for rate-function 
delay optimality in the many-channel many-user asymptotic regime and for throughput 
optimality in non-asymptotic settings. We then develop two classes of policies, 
called the Oldest Packets First (OPF) policies and the Maximum Weight in the Fluid 
limit (MWF) polices, that satisfy the sufficient condition for rate-function delay 
optimality and throughput optimality, respectively.

As discussed in the introduction, our ultimate goal is to \emph{develop 
low-complexity hybrid policies that are both rate-function delay-optimal 
and throughput-optimal.} However, it is unclear that, just because one 
policy is rate-function delay-optimal and another one is throughput-optimal, 
their combinations will necessarily yield the right hybrid policy that is 
optimal in terms of both throughput and delay. As we will discuss further 
at the beginning of Section~\ref{sec:hybrid}, our carefully chosen sufficient 
conditions possess some special features that allow us to construct low-complexity 
hybrid policies that are both rate-function delay-optimal and throughput-optimal.

%%%%%%%%%%%%%%%%%%%%%%%%%%%%%%%%%%%%%%%%%%%%%%%%%%%%%%%%%%%%%%%
\subsection{Rate-function Delay Optimality} \label{subsec:rfdo}
%%%%%%%%%%%%%%%%%%%%%%%%%%%%%%%%%%%%%%%%%%%%%%%%%%%%%%%%%%%%%%%
We start by presenting the main result of this section in the following 
theorem, which provides a sufficient condition for scheduling policies 
to be rate-function delay-optimal.

\begin{theorem}
\label{thm:suff-rf}
Under Assumptions~\ref{ass:arr_bound} and \ref{ass:arr_ld}, a scheduling 
policy $\mathbf{P}$ is rate-function delay-optimal if in any time-slot, 
policy $\mathbf{P}$ can serve the $k$ oldest packets in that time-slot 
for the largest possible value of $k \in \{1,2,\dots,n\}$.
\end{theorem}

% To prove Theorem~\ref{thm:suff-rf}, we will exploit a dominance property 
% (Lemma~\ref{lem:dom}) of the policies that satisfy the above sufficient 
% condition. Due to space constraint, we have provided the full proof with 
% all the details in our online technical report \cite{ji13a}. However, in 
% this paper we do provide an outline of the proof in Appendix~\ref{app:thm:suff-rf},
% and give the intuition behind it as follows. 
To prove Theorem~\ref{thm:suff-rf}, we will exploit a dominance property 
(Lemma~\ref{lem:dom}) of the policies that satisfy the above sufficient 
condition. We provide the proof of Theorem~\ref{thm:suff-rf} in 
Appendix~\ref{app:thm:suff-rf}, and give the intuition behind it as follows. 
First, it is easy to see that the First-come First-serve (FCFS) policy, which serves 
the \emph{oldest packets first}, is (sample-path) delay-optimal in a single-queue 
single-server system. Also, it is not hard to see that for a multi-queue multi-server 
system with \emph{full connectivity}, where all pairs of queue and server are connected, 
a policy that chooses to serve the oldest packets (over the whole system) first is 
delay-optimal. These motivate us to ask a natural and interesting question: \emph{if 
a policy chooses to serve the oldest packets first in a multi-queue multi-server 
system with time-varying and partial connectivity (as we consider in this paper), 
does it achieve rate-function delay optimality?} Note that in such a system, at most 
$n$ packets can be served in each time-slot. Hence, if in each time-slot a policy can 
serve all the $n$ oldest packets in the system (as in the case with full connectivity), 
this policy should yield optimal delay performance. However, due to the random 
connectivity between queues and servers, no policy may be able to do so. Hence, 
we propose a class of policies that choose to serve the $k$ oldest packets for 
the largest possible value of $k$. In other words, for any $k \in \{1,2,\dots,n\}$, 
if the $k$ oldest packets can be served by some scheduling policy, then our proposed 
policies will serve these $k$ packets too.

A similar, but less thorough, analysis was also carried out in 
\cite{sharma11,sharma11b}. There, the authors proposed the 
\textbf{Frame Based Scheduling (FBS)} policy, which aims to serve the oldest 
packets in each time-slot and can be viewed as an approximation of FCFS policy. 
The FBS policy serves packets in units of frames. With a given positive 
integer $h$ as the operating parameter, each frame is constructed such that: 
1) the difference of the arrival times of any two packets within a frame must 
be no greater than $h$; and 2) the total number of packets in each frame is 
no greater than $n_0=n-Lh$. In each time-slot, the packets arrived at the 
beginning of this time-slot are filled into the last frame until any of the 
above two conditions are violated, in which case a new frame will be opened. 
In any time-slot, the FBS policy serves the HOL frame that contains the oldest 
(up to $n_0$) packets with high probability for large $n$. As discussed at the
beginning of this section, it has been shown that the FBS policy with a properly 
chosen operating parameter $h$ is rate-function delay-optimal in some cases.

However, FBS may \emph{not} be rate-function delay-optimal in some other cases. 
Specifically, consider \emph{i.i.d.} Bernoulli arrivals with $L=1$. As pointed 
out in \cite{sharma11b}, the rate-function attained by the FBS policy is not
optimal in this scenario. We provide the intuition as follows. Suppose there are 
a total of $nt$ packet arrivals to the system in an interval of $t$ time-slots. 
It is easy to see that FBS needs at least $t+1$ time-slots to completely serve 
these packets since at most $n-Lh$ packets can be served by FBS in one time-slot. 
This could lead to a sub-optimal rate-function. To see this, consider the 
\textbf{perfect-matching} policy defined as follows. Let $\mathcal{Q}$ and 
$\mathcal{S}$ denote the set of queues and set of servers, respectively. In a 
time-slot $\tau$, let $\mathcal{C} \triangleq \{C_{i,j}(\tau): C_{i,j}(\tau)=1 \}$ 
denote the set of edges between $\mathcal{Q}$ and $\mathcal{S}$. Clearly, 
$G[\mathcal{Q} \cup \mathcal{S}, \mathcal{C}]$ forms a bipartite graph. If a 
perfect matching can be found in the bipartite graph $G[\mathcal{Q} \cup \mathcal{S}, 
\mathcal{C}]$, then the servers are allocated to serve the oldest packets in the
respective queues as determined by the perfect matching. Otherwise, none of servers 
will be allocated to the queues. It has been shown in \cite{bodas09} that in each 
time-slot, a perfect matching can be found with high probability for large $n$. 
Hence, in the case described above, the perfect-matching policy needs only $t$ 
time-slots to drain all these $nt$ packets with high probability for large $n$, 
while FBS is sub-optimal.

% \begin{proposition}
% \label{pro:pm}
% Suppose $L=1$. Then, for any fixed integer $b \ge 0$, the perfect-matching 
% policy attains a rate-function that is no smaller than $(b+1)I_X$, and is
% thus rate-function delay-optimal.  
% \end{proposition}

% We provide the proof in Appendix~\ref{app:pro:pm}.

On the other hand, the perfect-matching policy does not perform well in many other 
cases due to the fact that it cannot serve more than one packet from each queue in 
a time-slot. For example, consider the case where there are $L$ packets existing in 
$Q_1$ and the other queues are all empty. FBS can drain these packets within one 
time-slot with high probability, yet the perfect-matching policy needs at least $L$ 
time-slots.

The above discussions suggest that if we can find a policy that dominates both the 
FBS policy and the perfect-matching policy, there is a hope that this policy may be 
able to achieve the optimal rate-function in general. We will show in Lemma~\ref{lem:dom} 
that a policy that satisfies the sufficient condition in Theorem~\ref{thm:suff-rf}, 
indeed dominates both the FBS policy and the perfect-matching policy in a sample-path 
sense. 

In order to state the dominance property of Lemma~\ref{lem:dom} below, we consider 
the following versions of the FBS policy and the perfect-matching policy. Suppose 
that packet $p$ is the $x_p$-th arrival to the queue $Q_{q(p)}$ in time-slot $t_p$. 
Then, we define the weight of the packet $p$ in time-slot $t$ as $\hat{w}(p) = 
t - t_p + \frac {L+1-x_p} {L+1} + \frac {n+1-q(p)} {(L+1) (n+1)} $. For two 
packets $p_1$ and $p_2$, we say $p_1$ is older than $p_2$ if $\hat{w}(p_1) > 
\hat{w} (p_2)$. The above way of defining the weight ensures that among the 
packets that arrive at the same time, the priority is given to the packet that 
has an earlier order of arrival in each queue; and further, among the packets
(in different queues) with the same order of arrival, the priority is given to 
the packet that arrives to the queue with a smaller index.
For the FBS policy, we assume that the packets with a larger weight are filled 
to the frame with a higher priority when there are multiple packets arriving at 
the same time. While for the perfect-matching policy, we require that in time-slot
$t$, the perfect-matching policy only serves packets with the largest value of 
$t - t_p + \frac {L+1-x_p} {L+1}$. Under this version of the perfect-matching
policy, it is possible that a queue may not have any of its packets served even 
if a perfect-matching is found and a server is allocated to the queue. It should
be noted that the above versions of the FBS policy and the perfect-matching policy 
are used for analysis only. Next, we present the dominance property in the following 
lemma.

\begin{lemma}
\label{lem:dom}
Consider the versions of the FBS policy and the perfect-matching policy 
described above. Suppose that policy $\mathbf{P}$ satisfies the sufficient 
condition in Theorem~\ref{thm:suff-rf}. Then, for any given sample path, 
by the end of any time-slot $t$, policy $\mathbf{P}$ has served every 
packet that the FBS policy or the perfect-matching policy has served.
\end{lemma}

We prove Lemma~\ref{lem:dom} by contradiction, and provide the proof in 
Appendix~\ref{app:lem:dom}. Further, by using of this dominance property, 
and following a similar argument as in the rate-function delay analysis 
for FBS (Theorem~2 of \cite{sharma11b}), we prove Theorem~\ref{thm:suff-rf}. 
Specifically, we consider all the sample paths that lead to the delay-violation
event. There are different ways that the delay-violation event can occur, 
each of which has a corresponding rate-function for its probability of 
occurring. Large-deviations theory then tells us that the rate-function 
for delay violation is determined by the smallest rate-function among these 
possibilities (i.e., ``rare events occur in the most likely way".) An outline
of the proof for Theorem~\ref{thm:suff-rf} is provided in Appendix~\ref{app:thm:suff-rf}.

Next, we define a class of OPF policies as follows.

\begin{definition}
\label{def:rof}
A scheduling policy $\mathbf{P}$ is said to be in the class of \textbf{Oldest 
Packets First (OPF)} policies if policy $\mathbf{P}$ satisfies the sufficient 
condition in Theorem~\ref{thm:suff-rf}.
\end{definition}

Clearly, the class of OPF policies are all rate-function delay-optimal. We would
like to emphasize that the sufficient condition in Theorem~\ref{thm:suff-rf} is 
very easy to verify and can be readily used to design other rate-function delay-optimal 
policies. Specifically, Theorem~\ref{thm:suff-rf} enables us to identify a new 
rate-function delay-optimal policy, called the \textbf{DWM-$n$} policy, which
substantially reduces the complexity to $O(n^{2.5} \log n)$. This in turn allows 
us to design \emph{low-complexity} hybrid scheduling policies that are both 
\emph{throughput-optimal} and \emph{rate-function delay-optimal} (in Section~\ref{sec:hybrid}).

Now, we review the \textbf{Delay Weighted Matching (DWM)} policy proposed in 
\cite{sharma11,sharma11b}. DWM operates in the following way. In each time-slot 
$t$, define the weight of the $l$-th packet of $Q_i$ as $Z_{i,l}(t)$, i.e., the 
delay of this packet at the beginning of time-slot $t$, which is measured since 
the time when this packet arrived to queue $Q_i$ until time-slot $t$. Then, 
construct a bipartite graph $G[X \cup Y,E]$ such that the vertices in $X$ 
correspond to the $n$ oldest packets from each of the $n$ queues and $Y$ is the 
set of all servers. Thus, $|X|=n^2$ and $|Y|=n$. Let $X_i \subseteq X$ be the set 
of packets from queue $Q_i$. If queue $Q_i$ is connected to server $S_j$, then for 
each packet $x \in X_i$, there is an edge between $x$ and $S_j$ in graph $G$ and 
the weight of this edge is set to the weight of packet $x$. The schedule is then 
determined by a maximum-weight matching over $G$. Clearly, DWM maximizes the sum 
of the delays of the packets scheduled.

It has been shown in \cite{sharma11,sharma11b} that the DWM policy is 
rate-function delay-optimal in some cases. However, it is unclear whether 
it is delay-optimal in general. \emph{We would like to highlight that our
proposed sufficient condition in Theorem~\ref{thm:suff-rf} allows us to 
show that a slightly modified version of the DWM policy is rate-function 
delay-optimal in general (under an additional mild assumption - the second 
part of Assumption~\ref{ass:arr_bound}).} Specifically, in the modified 
version of the DWM policy, we assign the weight of a packet $p$ as $\hat{w}(p)$ 
instead of its delay only. Then, by simply duplicating the proof of Lemma~7 
in \cite{sharma11b}, we can show that the modified version of the DWM policy 
is an OPF policy and is thus rate-function delay-optimal. 

However, the DWM policy still suffers from a high complexity, which renders 
it impractical. Specifically, DWM has a complexity of $O(n^5)$, since the 
complexity of finding a maximum-weight matching \cite{fredman87} over a 
bipartite graph $G[V,E]$ is $O(|V| |E| + |V|^2 \log |V|)$ in general, and 
the bipartite graph constructed by DWM has $|V|=O(n^2)$ and $|E|=O(n^3)$. 

To overcome the high-complexity issue, we develop a simpler policy that is 
also in the class of the OPF policies (and is thus rate-function delay-optimal), 
but has a much lower complexity of $O(n^{2.5} \log n)$. The new policy is 
called the \textbf{DWM-$n$} policy due to the high-level similarity with DWM. 
However, it exhibits critical differences when picking packets to construct the 
bipartite graph $G[X \cup Y, E]$ and finding the maximum-weight matching over 
$G$. The differences are as follows:
\begin{enumerate}
\item In each time-slot, instead of considering the $n$ oldest packets 
from each queue (and thus $n^2$ packets in total) as in DWM, DWM-$n$ 
considers only the $n$ oldest packets in the whole system. Hence, the 
bipartite graph constructed by DWM-$n$ has $|X|=n$ and $|Y|=n$. 
\item The rest of the operations of DWM-$n$ are similar to that of DWM, 
i.e., the schedule is determined by a maximum-weight matching over $G$, 
except that DWM-$n$ finds a maximum-weight matching based on the 
\emph{vertex} weights. Such a maximum-weight matching is also called 
Maximum Vertex-weighted Matching (MVM) \cite{spencer84,gupta09}. 
Specifically, the weight of each vertex $p \in X$ is set to $\hat{w}(p)$ 
(i.e., the weight of the corresponding packet $p$), and the weight of 
each vertex in the set $Y$ is set to 0.
\end{enumerate}

In the following proposition, we show that the DWM-$n$ policy is 
rate-function delay-optimal and has a low complexity. 
\begin{proposition}
\label{pro:dwmn}
The DWM-$n$ policy is an OPF policy, and is thus rate-function delay-optimal
under Assumptions~\ref{ass:arr_bound} and \ref{ass:arr_ld}. Further, the 
DWM-$n$ policy has a low complexity of $O(n^{2.5} \log n)$.
\end{proposition}

We provide the proof in Appendix~\ref{app:pro:dwmn}. The fact that the DWM-$n$ 
policy is an OPF policy follows from a property of MVM \cite{spencer84} that 
if there exists a matching that matches all of the $k$ heaviest vertices, then 
any MVM matches all of the $k$ heaviest vertices as well. The low complexity 
of DWM-$n$ follows immediately from the fact that DWM-$n$ reduces the number 
of packets under consideration ($n$ packets in total), and that an MVM in an 
$n \times n$ bipartite graph can be found in $O(n^{2.5} \log n)$ time 
\cite{spencer84}. Note that even if the DWM policy adopts MVM when determining 
the schedule, its complexity can only be reduced to $O(n^4 \log n)$. 

Although the DWM-$n$ policy achieves rate-function delay optimality with a 
low complexity, it may not be throughput-optimal in general. This is because 
the DWM-$n$ policy considers only the $n$ oldest packets in the system. It is 
likely that certain servers may not be connected to any of the queues that 
contain these $n$ packets, which results in the server being idle and is thus 
a waste of service. Hence, DWM-$n$ is a lazy policy. In fact, we can construct 
a simple counter-example to show that the DWM-$n$ policy is, in general, not 
throughput-optimal as stated in Proposition~\ref{pro:unstable}.

\begin{proposition}
\label{pro:unstable}
The DWM-$n$ policy is not throughput-optimal in general.
\end{proposition}

We prove Proposition~\ref{pro:unstable} by constructing a special arrival 
pattern that forces certain servers to be idle, even when they can serve 
some of the queues. We provide the proof in Appendix~\ref{app:pro:unstable}. 
Proposition~\ref{pro:unstable} suggests that a rate-function delay-optimal 
policy may not have good throughput performance (for a fixed $n$). This may 
appear counter-intuitive at the first glance. However, it should be noted 
that the rate-function delay optimality is studied in the asymptotic regime, 
i.e., when $n$ goes to infinity. Although the convergence of the rate-function 
is typically fast, the throughput performance may be poor for small to moderate 
values of $n$. Our simulation results (Fig.~\ref{fig:delay} in Section~\ref{sec:sim}) 
will provide further evidence of this.

%%%%%%%%%%%%%%%%%%%%%%%%%%%%%%%%%%%%%%%%%%%%%%%%%%%%
\subsection{Throughput Optimality} \label{subsec:to}
%%%%%%%%%%%%%%%%%%%%%%%%%%%%%%%%%%%%%%%%%%%%%%%%%%%%
In this section, we present a sufficient condition for throughput optimality 
in very general non-asymptotic settings. 

Recall that $Q_i(t)$ denotes the length of queue $Q_i$ at the beginning
of time-slot $t$ immediately after packet arrivals, $Z_{i,l}(t)$ denotes 
the delay of the $l$-th packet of $Q_i$ at the beginning of time-slot $t$, 
$W_i(t) = Z_{i,1}(t)$ denotes the HOL packet delay of $Q_i$ at the beginning 
of time-slot $t$, and $C_{i,j}(t)$ denotes the connectivity between $Q_i$ 
and $S_j$ in time-slot $t$. Let $\mS_j(t)$ denote the set of queues being 
connected to server $S_j$ in time-slot $t$, i.e., $\mS_j(t) = \{1 \le i \le 
n ~|~ C_{i,j}(t)=1 \}$, and let $\Gamma_j(t)$ denote the subset of queues 
in $\mS_j(t)$ that have the largest weight in time-slot $t$, i.e., $\Gamma_j(t) 
\triangleq \{i \in \mS_j(t) ~|~ W_i(t) = \max_{l \in \mS_j(t)} W_l(t) \}$. 
We now present the main result of this section. 

\begin{theorem}
\label{thm:suff-to}
Let $i(j,t)$ be the index of the queue that is served by server $S_j$ in 
time-slot $t$, under a scheduling policy $\mathbf{P}$. Under 
Assumption~\ref{ass:arr_slln}, policy $\mathbf{P}$ is throughput-optimal
if there exists a constant $M>0$ such that, in any time-slot $t$ and for 
all $j \in \{1,2,\dots,n \}$, queue $Q_{i(j,t)}$ satisfies that 
$W_{i(j,t)}(t) \ge Z_{r,M}(t)$ for all $r \in \Gamma_j(t)$ such that 
$Q_r(t) \ge M$.
\end{theorem} 

We prove Theorem~\ref{thm:suff-to} using fluid limit techniques \cite{dai95,
andrews04}, and provide the proof in Appendix~\ref{app:thm:suff-to}. 
% We prove Theorem~\ref{thm:suff-to} using fluid limit techniques \cite{dai95,andrews04}
% and standard Lyapunov argument. Due to space constraint, we provide the proof in our 
% online technical report \cite{ji13a}.
The condition in Theorem~\ref{thm:suff-to} means the following: In each time-slot, 
\emph{each server chooses to serve a queue with HOL packet delay no less than 
the delay of the $M$-th packet in the queue with the largest HOL delay (among 
the queues connected to the server)}; if this queue (with the largest HOL delay) 
has less than $M$ packets, then the server may choose to serve any queue. 

It is well-known that the MaxWeight Scheduling (MWS) policy \cite{tassiulas92,
lin06c,georgiadis06,andrews04,eryilmaz05,ji13c} that maximizes the weighted sum 
of the rates, where the weights are queue lengths or delays, is throughput-optimal 
in very general settings, including the multi-channel system that we consider 
in this paper. The intuition behind Theorem~\ref{thm:suff-to} is that to achieve 
throughput optimality in our multi-channel system, it is sufficient for each 
server to choose a connected queue with a large enough weight such that this 
queue has the largest weight in the fluid limit. This relaxes the condition that 
each server has to find a queue with the largest weight in the original system, 
and thus significantly expands the set of known throughput-optimal policies.

Next, we define the class of Maximum Weight in the Fluid limit (MWF) policies 
as follows.

\begin{definition}
A policy $\mathbf{P}$ is said to be in the class of \textbf{Maximum Weight 
in the Fluid limit (MWF)} policies if policy $\mathbf{P}$ satisfies the 
sufficient condition in Theorem~\ref{thm:suff-to}.
\end{definition}

Clearly, the class of MWF policies are all throughput-optimal. It is claimed 
in \cite{sharma11,sharma11b} that the DWM policy is throughput-optimal, yet 
the throughput optimality was not explicitly proved there. For completeness, 
we state the following proposition on throughput optimality of the DWM policy, 
and provide its proof in Appendix~\ref{app:pro:dwm-to}.
% and provide its proof in our online technical report \cite{ji13a}.

\begin{proposition}
\label{pro:dwm-to}
The DWM policy is an MWF policy, and is thus throughput-optimal under 
Assumption~\ref{ass:arr_slln}.
\end{proposition}

Next, we study a simple extension of the delay-based MaxWeight 
policy \cite{andrews04,eryilmaz05,ji13c} that is throughput-optimal 
in our multi-channel system. 

\noindent {\bf Delay-based MaxWeight Scheduling (D-MWS) policy:}
In each time-slot $t$, the scheduler allocates each server $S_j$ 
to serve queue $Q_{i(j,t)}$ such that $i(j,t) = \min \{i ~|~ i \in 
\Gamma_j(t) \}$. In other words, each server chooses to serve a 
queue that has the largest HOL delay (among all the queues connected 
to this server), breaking ties by picking the one with the smallest 
index if there are multiple such queues.

It is easy to see that D-MWS is an MWF policy and is thus throughput-optimal. 
Also, it is worth noting that D-MWS has a low complexity of $O(n^2)$ in our 
mutli-channel system. However, we can show that D-MWS suffers from poor delay 
performance. Specifically, we show in the following proposition that under 
D-MWS, the probability that the largest HOL delay exceeds any fixed threshold, 
is at least a constant, even if $n$ is large. This results in a zero rate-function.

\begin{proposition}
\label{pro:dmws-rf}
Consider \emph{i.i.d.} Bernoulli arrivals, i.e., in each time-slot, 
and for each user, there is a packet arrival with probability 
$p$, and no arrivals otherwise. By allocating servers to queues 
according to D-MWS, we have that
\begin{equation}
\limsup_{n \rightarrow \infty} \frac {-1} {n} 
\log \Prob \left( W(0)>b \right) = 0,
\end{equation}
for any fixed integer $b\ge0$.
\end{proposition}

We provide the proof in Appendix~\ref{app:pro:dmws-rf}, and explain the 
intuition behind it in the following. 
% Due to space constraint, we provide the proof in our online technical 
% report \cite{ji13a} and explain the intuition behind it in the following. 
Note that under D-MWS, each server chooses to serve a connected queue 
having the largest weight without accounting for the decisions of the 
other servers. This way of allocating servers may incur an unbalanced 
schedule such that in each time-slot, with high probability, only a 
small fraction of the queues ($O(\log n)$ out of $n$ queues) get served, 
while the number of queues having arrivals is much larger ($O(n)$). 
This then leads to poor delay performance. By an argument similar to 
that in Theorem~3 of \cite{bodas10} (where the authors show that the 
Queue-length-based MaxWeight Scheduling (Q-MWS) policy results in a 
zero queue-length rate-function), we can show that under D-MWS, the 
delay-violation event occurs with at least a constant probability for 
any fixed threshold even if $n$ is large. 

We conclude this section with a summary of the scheduling policies 
proposed and/or discussed in this section. The FBS policy is a good 
policy that is useful for the rate-function delay analysis of other 
policies, yet it is neither throughput-optimal nor rate-function 
delay-optimal in general. Although (the modified version of) the DWM 
policy is both throughput-optimal and rate-function delay-optimal, 
it yields an impractically high complexity. Our analysis shows that 
our proposed the DWM-$n$ policy is rate-function delay-optimal and 
substantially reduces the complexity to $O(n^{2.5} \log n)$, but it 
is not throughput-optimal. Further, we show that a simple throughput-optimal 
policy, the D-MWS policy, suffers from a zero rate-function.

%%%%%%%%%%%%%%%%%%%%%%%%%%%%%%%%%%%%%%%%%%%%
\section{Hybrid Policies} \label{sec:hybrid}
%%%%%%%%%%%%%%%%%%%%%%%%%%%%%%%%%%%%%%%%%%%%
It is clear from the previous section that a policy that satisfies the 
sufficient conditions in Theorems~\ref{thm:suff-rf} and \ref{thm:suff-to} 
is both throughput-optimal and rate-function delay-optimal. It remains 
however to find such a policy with a \emph{low complexity}. Interestingly, 
our carefully chosen sufficient conditions possess the following special 
features, which allow us to construct a \emph{low-complexity hybrid policy 
that is both rate-function delay-optimal and throughput-optimal}:
\begin{itemize}
\item The sufficient condition for throughput optimality has a decoupling 
feature, in the sense that the condition can be separately verified for 
each individual server.

\item The sufficient condition for rate-function delay optimality guarantees 
not only rate-function delay optimality itself, \emph{but also that all 
scheduled servers for the $n$ oldest packets satisfy the sufficient condition 
for throughput optimality}.
\end{itemize}
Hence, by exploiting the above useful features of our sufficient conditions, 
we can now develop a class of two-stage hybrid OPF-MWF policies that runs an 
OPF policy (focusing on the $n$ oldest packets only) in stage~1, and runs an 
MWF policy in stage~2 over the remaining servers (that are not allocated in 
stage~1) only. We will then show that all policies in this class of hybrid 
OPF-MWF policies are both rate-function delay-optimal and throughput-optimal. 
In particular, we can find simple OPF-MWF policies with a low complexity
$O(n^{2.5} \log n)$.

We now formally define the class of two-stage hybrid OPF-MWF policies.

\begin{definition}
\label{def:hybrid}
A scheduling policy $\mathbf{P}$ is said to be in the class of \textbf{hybrid 
OPF-MWF} policies, if the following conditions are satisfied under policy 
$\mathbf{P}$: In each time-slot $t$, there are two stages:
\begin{enumerate}
\item in stage~1, it runs an OPF policy over the $n$ oldest packets only;
\item in stage~2, let $R(t)$ denote the set of servers that are not 
allocated by the OPF policy in stage~1, and let $i(j,t)$ be the index 
of the queue that is matched by server $S_j$ for $j \in R(t)$ in stage~2. 
There exists a constant $M>0$ such that in any time-slot $t$ and for 
all $j \in R(t)$, queue $Q_{i(j,t)}$ satisfies that $W_{i(j,t)}(t) \ge 
Z_{r,M}(t)$ for all $r \in \Gamma_j(t)$ such that $Q_r(t) \ge M$. In 
other words, it runs an MWF policy over the system with the remaining 
servers and packets.
\end{enumerate}
\end{definition}

In the following theorem, we show that the class of OPF-MWF policies 
are both rate-function delay-optimal and throughput-optimal.

\begin{theorem}
\label{thm:hybrid}
Any hybrid OPF-MWF policy is rate-function delay-optimal under 
Assumptions~\ref{ass:arr_bound} and \ref{ass:arr_ld}, and is 
throughput-optimal under Assumption~\ref{ass:arr_slln}.
\end{theorem}

We provide the proof in Appendix~\ref{app:thm:hybrid}, and give the 
intuition behind it as follows.
In stage~1, an OPF policy not only guarantees rate-function 
delay optimality, \emph{but also satisfies the sufficient condition 
for throughput optimality for all allocated servers in this stage}. 
Note that the allocated servers and packets in stage~1 will not be considered 
in stage~2. In stage~2, \emph{we run an MWF policy for the remaining servers 
and packets only.} Hence, it ensures that the sufficient condition for 
throughput optimality is satisfied for the remaining servers as well. Since 
the allocated servers and packets in stage~1 are not touched in stage~2, the 
satisfaction of the sufficient condition for delay optimality is not perturbed, 
and the sufficient condition for throughput optimality is also satisfied.

% \begin{itemize}
% \item In stage~1, an OPF policy not only guarantees rate-function 
% delay optimality, \emph{but also satisfies the sufficient condition 
% for throughput optimality for all allocated servers in this stage}. 

% \item The allocated servers and packets in stage~1 will not be considered 
% in stage~2. In stage~2, \emph{we run an MWF policy for the remaining servers 
% and packets only.} Hence, it ensures that the sufficient condition for 
% throughput optimality is satisfied for the remaining servers as well. Since 
% the allocated servers and packets in stage~1 are not touched in stage~2, the 
% satisfaction of the sufficient condition for delay optimality is not perturbed, 
% and the sufficient condition for throughput optimality is also satisfied.
% \end{itemize} 

We note that the idea of combining different policies into (heuristic) hybrid 
policies to improve the overall performance, is not new. However, our goal in 
this paper is to achieve \emph{provable} optimality in terms of both throughput 
and delay. Hence, the task of designing the right hybrid policy becomes much more 
challenging. Further, it is not necessary that all combinations of the OPF and 
MWF policies lead to desired hybrid policies. For example, it is unclear that 
the sufficient condition for throughput optimality can be satisfied if instead, 
we run an MWF policy in stage~1 and do post-processing by applying an OPF policy 
in stage~2. In this case, because the servers allocated by an MWF policy in stage~1 
can be reallocated in stage~2, the sufficient condition for throughput optimality 
may not hold any more. In contrast, our solutions exploit the special features of 
our carefully chosen sufficient conditions, and intelligently combine different 
policies in a right way, to achieve the optimal performance for both throughput 
and delay.

There are still many policies in the class of hybrid OPF-MWF policies. 
In the following, as an example, we show that the DWM-$n$ policy combined 
with the D-MWS policy yields an $O(n^{2.5} \log n)$-complexity hybrid OPF-MWF 
policy that is both throughput-optimal and rate-function delay-optimal. 
Let this policy be called \textbf{DWM-$n$-MWS} policy. Then, we present
the main result of this paper in the following theorem.

\begin{theorem}
\label{thm:dwm-mws}
DWM-$n$-MWS policy is in the class of hybrid OPF-MWF policies, 
and is thus both throughput-optimal and rate-function delay-optimal. 
Further, DWM-$n$-MWS policy has a complexity of $O(n^{2.5} \log n)$. 
\end{theorem}

To show that DWM-$n$-MWS is a hybrid OPF-MWF policy, it suffices to 
show that Condition 2) of Definition~\ref{def:hybrid} is satisfied. 
We provide the proof in Appendix~\ref{app:thm:dwm-mws}.

%%%%%%%%%%%%%%%%%%%%%%%%%%%%%%%%%%%%%%%%%%%%
\section{Simulation Results} \label{sec:sim}
%%%%%%%%%%%%%%%%%%%%%%%%%%%%%%%%%%%%%%%%%%%%
In this section, we conduct simulations to compare the performance of the 
scheduling policies proposed or discussed in this paper, where the Hybrid 
policy we consider is DWM-$n$-MWS policy. We also compare the delay 
performance of our proposed policies along with two $O(n^2)$-complexity queue-length-based 
policies (i.e., using queue lengths instead of delays to calculate weights 
when making scheduling decisions): Queue-based Server-Side-Greedy (Q-SSG) 
and Q-MWS, which have been studied in \cite{bodas09,bodas10}. We implement 
and simulate these policies in Java, and compare the empirical probabilities 
that the largest HOL delay in the system in any given time-slot exceeds a 
constant $b$, i.e., $\Prob (W(0)>b)$.

For the arrival processes, we consider bursty arrivals that are driven by a 
two-state Markov chain and are thus correlated over time. (We obtained similar
results for \emph{i.i.d.} arrivals over time, but omit them here due to space 
constraints.) We adopt the same parameter settings as in \cite{sharma11,sharma11b}. 
For each user, there are 5 packet-arrivals when the Markov chain is in state~1, 
and no arrivals when the Markov chain is in state~2. The transition probability 
of the Markov chain is given by the matrix $[0.5,0.5;0.1,0.9]$,
% $\begin{bmatrix}
% 0.5 & 0.5 \\
% 0.1 & 0.9
% \end{bmatrix}$,
and the state transitions occur at the end of each time-slot. The arrivals 
for each user are correlated over time, but they are independent across users. 
For the channel model, we first assume \emph{i.i.d.} ON-OFF channels (as in 
Assumption~\ref{ass:channel}) and set $q=0.75$, and later consider more general 
scenarios with heterogeneous users and bursty channels that are correlated over 
time. We run simulations for a system with $n \in \{ 10,20,\dots,100 \}$. The 
simulation period lasts for $10^7$ time-slots for each policy and each system.

% \begin{figure*}
% \centering
% \parbox{0.3\linewidth}{
% \epsfig{file=MC_b_2.eps,width=1\linewidth}
% \caption{Performance comparison of different scheduling policies in the case with 
% homogeneous \emph{i.i.d.} channels, for delay threshold $b=2$.}
% \label{fig:dr}}
% \qquad
% \begin{minipage}{0.3\linewidth}
% \epsfig{file=MC_N_10.eps,width=1\linewidth}
% \caption{Performance comparison of different scheduling policies in the case with 
% homogeneous \emph{i.i.d.} channels, for $n=10$ channels/users.}
% \label{fig:delay}
% \end{minipage}
% \qquad
% \begin{minipage}{0.3\linewidth}
% \epsfig{file=MC_b_heter_2.eps,width=1\linewidth}
% \caption{Performance comparison of different scheduling policies 
% in the case with Markov-chain driven heterogeneous channels, for delay threshold $b=2$.}
% \label{fig:heter}
% \end{minipage}
% \end{figure*}

% \begin{figure}[t]
% \centering
% \subfigure[$b=2$]{\epsfig{file=MC_b_2.eps,width=0.3\linewidth} \label{fig:dr}}
% \subfigure[$n=10$]{\epsfig{file=MC_N_10.eps,width=0.3\linewidth} \label{fig:delay}}
% \caption{Performance comparison of different scheduling policies in the case with 
% homogeneous \emph{i.i.d.} channels.}
% \label{fig:simple}
% \end{figure}

\begin{figure}[t]
\centering
\epsfig{file=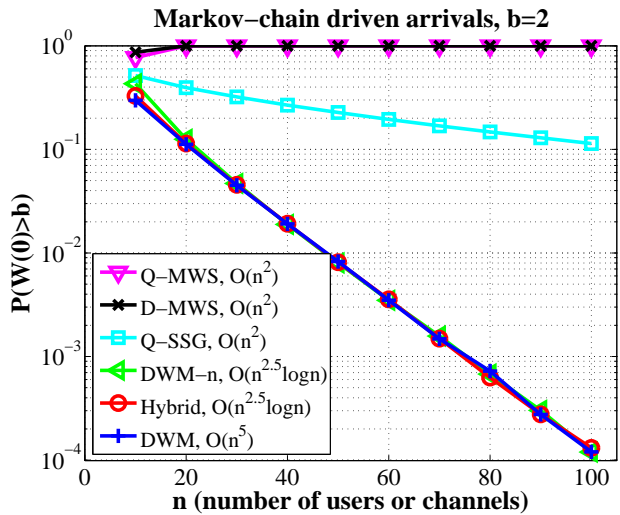,width=0.5\linewidth}
\caption{Performance comparison of different scheduling policies in the case with 
homogeneous \emph{i.i.d.} channels, for delay threshold $b=2$.}
\label{fig:dr}
\end{figure}

\begin{figure}[t]
\centering
\epsfig{file=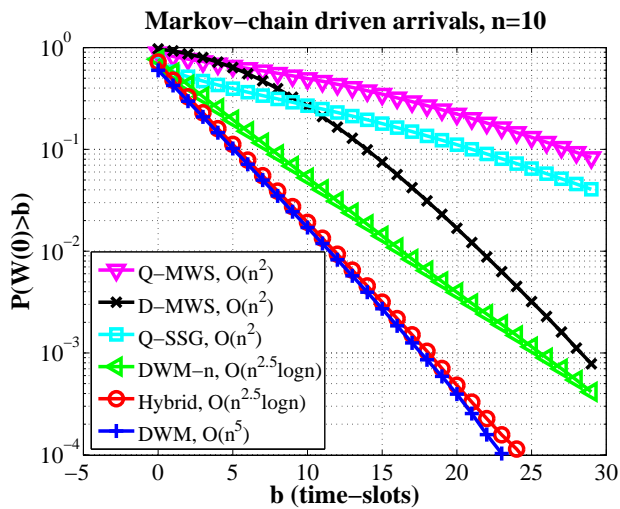,width=0.5\linewidth}
\caption{Performance comparison of different scheduling policies in the case with 
homogeneous \emph{i.i.d.} channels, for $n=10$ channels/users.}
\label{fig:delay}
\end{figure}

\begin{figure}[t]
\centering
\epsfig{file=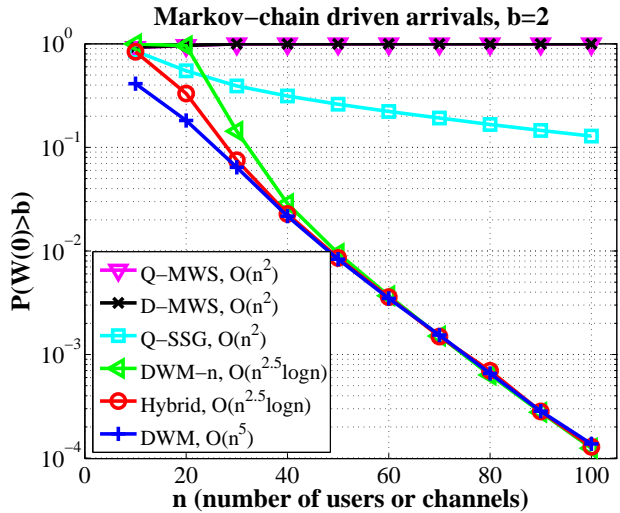,width=0.5\linewidth}
\caption{Performance comparison of different scheduling policies 
in the case with Markov-chain driven heterogeneous channels, for delay threshold $b=2$.}
\label{fig:heter}
\end{figure}

The results are summarized in Figs.~\ref{fig:dr} and \ref{fig:delay}, 
where the complexity of each policy is labeled. In order to compare 
the rate-function $I(b)$ as defined in Eq.~(\ref{eq:rf}), we plot the 
probability over the number of channels or users, i.e., $n$, for a 
fixed value of threshold $b$. In Fig.~\ref{fig:dr}, we compare the 
rate-function $I(b)$ of different scheduling policies for $b=2$. The 
negative of the slope of each curve can be viewed as the rate-function 
for the corresponding policy. From Fig.~\ref{fig:dr}, we observe that 
the Hybrid and DWM-$n$ policies perform closely to DWM, and that D-MWS 
and Q-MWS have a zero rate-function, which supports our analytical 
results. Further, the results show that the delay-based policies (DWM, 
DWM-$n$ and Hybrid) consistently outperform Q-SSG in terms of delay 
performance, despite that it has been shown through simulations that 
Q-SSG performs closely to a rate-function (queue-length) optimal policy 
\cite{bodas09,bodas10}. This provides further evidence of the fact that 
good queue-length performance does not necessarily translate to good 
delay performance.

We also plot the probability over delay threshold $b$ as in \cite{bodas09,
bodas10,bodas11a,sharma11,sharma11b} to investigate the performance of different 
policies when $n$ is small. In Fig.~\ref{fig:delay}, we report the results 
for $n=10$ and $b \in \{0,1,2,\dots,29\}$. From Fig.~\ref{fig:delay}, we 
observe that the Hybrid policy consistently performs closely to DWM for 
almost all values of $b$ that we consider, while DWM-$n$ is worse than DWM. 
This is because DWM-$n$ may not schedule all the servers, and the probability 
that some of the servers are kept idle can be significant when $n$ is small. 

Finally, we evaluate the performance of different scheduling policies 
in more realistic scenarios, where users are \emph{heterogeneous} and 
channels are \emph{correlated over time}. Specifically, we consider 
channels that can be modeled by a two-state Markov chain, where the 
channel is ``ON" when the Markov chain is in state~1, and is ``OFF" 
when the Markov chain is in state~2. This type of channel model can 
be viewed as a special case of the Gilbert Elliot model that is widely 
used for describing bursty channels. We assume that there are two 
classes of users: users with an odd index are called \emph{near-users}, 
and users with an even index are called \emph{far-users}. Different 
classes of users see different channel conditions: near-users see 
better channel condition, and far-users see worse channel condition. 
We assume that the transition probability matrices of channels for 
near-users and far-users are $[0.833, 0.167; 0.5 , 0.5]$ and $[0.5, 
0.5; 0.167, 0.833]$, respectively. The arrival processes are assumed 
to be the same as in the previous case. Also, the delay requirements 
are assumed to be the same for different classes of users, i.e., we 
still consider the probability that the largest HOL delay exceeds a 
fixed threshold, without distinguishing different classes of users.  

The results are summarized in Fig.~\ref{fig:heter}. We observe similar 
results as in the previous case, where channels are \emph{i.i.d.} in 
time. In particular, our low-complexity policies (DWM-$n$ and Hybrid) 
again perform closely to DWM, in terms of rate-function, although the
delay-violation probability is a bit smaller under DWM when $n$ is not 
large (i.e., $n<50$), which is expected. Note that in this scenario, 
rate-function delay-optimal policies are \emph{not} known yet. For 
future work, it would be interesting to explore whether our proposed 
policies can achieve optimality of both throughput and delay in more 
general scenarios.

%%%%%%%%%%%%%%%%%%%%%%%%%%%%%%%%%%%%%
\section{Conclusion} \label{sec:con}
%%%%%%%%%%%%%%%%%%%%%%%%%%%%%%%%%%%%%
In this paper, we addressed the question of designing low-complexity 
scheduling policies that provide optimal performance of both throughput
and delay in multi-channel systems. We derived simple and easy-to-verify 
sufficient conditions for throughput optimality and rate-function delay 
optimality, which allowed us to later develop a class of low-complexity 
hybrid policies that simultaneously achieve both throughput optimality 
and rate-function delay optimality.

Our work in this paper leads to many interesting questions that are worth 
exploring in the future. It would be interesting to know if one can further 
relax the sufficient conditions, and design even simpler policies that can 
provide optimal performance for both throughput and delay. Further, it would 
be worthwhile to analytically characterize the fundamental trade-off between 
performance and complexity. 

Further, it is important to investigate the scheduling problem in 
more realistic scenarios, e.g., accounting for more general multi-rate 
channels that are correlated over time, rather than \emph{i.i.d.} ON-OFF 
channels, and heterogeneous users with different statistics as well as 
different delay requirements. Our hope is to find efficient schedulers 
that can guarantee a nontrivial lower bound of the optimal rate-function, 
if it is too hard to achieve (or prove) the optimal delay performance 
itself in more general scenarios.

Finally, it is interesting and important for us to understand the delay 
performance beyond rate-function optimality as we considered in this paper. 
The log-asymptotic results from the large-deviations analysis may not suffice, 
since they do not account for the pre-factor of the delay-violation probability. 
Therefore, a very important direction is to analyze and understand the exact 
delay asymptotics as well as the mean delay performance.
%\newpage

\appendices 
%for IEEE

%\appendix
%for ACM

%%%%%%%%%%%%%%%%%%%%%%%%%%%%%%%%%%%%%%%%%%%%%%%%%%%%%%%%%%%%%%%%%%%
\section{The Optimal Throughput Region $\Lambda^{*}$} \label{app:otr}
%%%%%%%%%%%%%%%%%%%%%%%%%%%%%%%%%%%%%%%%%%%%%%%%%%%%%%%%%%%%%%%%%%%
We can characterize the optimal throughput region $\Lambda^{*}$ of our 
multi-channel systems in a similar manner to that of single-channel 
systems in \cite{andrews04}.

We start with discussions for a single-channel system with $n$ users
in a more general setting. Specifically, suppose that there is a finite 
set $\mM=\{1,2,\dots,|\mM|\}$ of global server states (where the server 
state accounts for the state of the links between the server and all 
users). For each state $m \in \mM$, there is an associated service rate 
vector $r^m=[r^m_i, 1 \le i \le n]$, where $r^m_i$ is the maximum number 
of packets that can be transmitted to $Q_i$ when the server is in state 
$m$ (under Assumption~\ref{ass:channel} of ON-OFF channels, we have $r^m_i 
\in \{0,1\}$ for all $m$ and $i$). We assume that the random channel state 
process is a stationary and ergodic discrete-time Markov chain within the 
state space $\mM$. We let $\pi=[\pi_m, m \in \mM]$ denote the stationary 
distribution of this Markov chain, where $\pi_m > 0$ for all $m \in \mM$.

As in \cite{andrews04}, consider a \emph{Static Service Split (SSS)}
policy, associated with an $|\mM| \times n$ stochastic matrix $\phi
=[\phi_{m,i},m \in \mM, 1 \le i \le n]$, where $\phi_{m,i} \ge 0$ for 
all $m$ and $i$, and $\sum_{1 \le i \le n} \phi_{m,i} = 1$ for every 
$m$. Under the SSS policy, the server chooses to serve $Q_i$ with 
probability $\phi_{m,i}$ when the server is in state $m$. Clearly, 
the (long-term average) service rate vector can be represented by 
$\nu=[\nu_i, 1 \le i \le n]=\nu(\phi)$, where $\nu_i=\sum_{m \in 
\mM} \pi_m \phi_{m,i} r^m_i$. Then, the set of all feasible (long-term 
average) service rate vector can be represented as 
\[
\Rate=\{ \nu ~|~ \nu=\nu(\phi) ~\text{for some stochastic matrices}~ \phi \}.
\]
Hence, the optimal throughput region can be represented as 
\[
\Lambda^{*}=\{\lambda ~|~ \lambda \le \nu ~\text{for some vector}~ \nu \in \Rate \}.
\]

Now, consider a multi-channel system with $n$ orthogonal channels. 
Let $\mu_{i,j}$ denote the feasible (long-term average) service rate 
that can be allocated to queue $Q_i$ from server $S_j$, and let the 
vector $\mu_j=[\mu_{i,j}, 1 \le i \le n]$ denote a feasible service 
rate allocation by server $S_j$. For each server $S_j$, the set of 
all such feasible vectors $\mu_j$ is denoted by $\Rate_j$. Note that 
the characterization of $\Rate_j$ has already accounted for the 
time-varying channel-states. Let $\mu=[\mu_j, 1 \le j \le n]$ denote 
a feasible service rate matrix, and the set of all such feasible 
matrices $\mu$ can be represented as $\Rate=\Rate_1 \times \Rate_2 
\times \dots \times \Rate_n$. Hence, the optimal throughput region 
$\Lambda^*$ can be represented as
\[
\Lambda^* = \{\lambda ~|~  \lambda_i \le \sum_{j=1}^n \mu_{i,j} 
~\text{for all}~ i, ~\text{for some matrix}~ \mu \in \Rate\}.
\]

Note that our multi-channel system under Assumption~\ref{ass:channel} 
of ON-OFF channel model is a special case of the above scenario.

%%%%%%%%%%%%%%%%%%%%%%%%%%%%%%%%%%%%%%%%%%%%%%%%%%%%%%%%%%%
\section{Proof of Theorem~\ref{thm:ub}}  \label{app:thm:ub}
%%%%%%%%%%%%%%%%%%%%%%%%%%%%%%%%%%%%%%%%%%%%%%%%%%%%%%%%%%%
We begin with stating an important property of $I_A(t,x)$ in the following 
lemma, which will be used in deriving the upper bound in Theorem~\ref{thm:ub}. 
Recall that we define the quantity $I^+_A(t,x) \triangleq \lim_{y \rightarrow x^+} I_A(t,y)$. 
\begin{lemma}
\label{lem:IAtx}
Suppose $L>1$. For any given integer $t>0$, and for all $x \in [0,(L-1)t)$, 
the limit $I^+_A(t,x)=\lim_{y \rightarrow x^+} I_A(t,y)$ exists and we have 
$I_A(t,x)=I^+_A(t,x)$.
\end{lemma}

\begin{proof}
Consider any given integer $t>0$. First, note that the total number of 
packet arrivals to the system during an interval of $t$ time-slots cannot 
exceed $Lnt$. Hence, we only need to consider $I_A(t,x)$ defined on 
$[0,(L-1)t]$. By the second part of Assumption~\ref{ass:arr_bound}, it 
is easy to see that $I_A(t,x)$ must be finite in $[0,(L-1)t)$. Note that 
$I_A(t,x)$ is a supremum (over $\theta$) of linear functions (of $x$). 
Hence, $I_A(t,x)$ is a convex function (of $x$), and is thus continuous 
on $(0,(L-1)t)$ (i.e., the interior of $[0,(L-1)t]$) \cite[Pg. 68]{boyd04}. 
Further, it is easy to see that $I_A(t,x)$ is 
monotone (non-decreasing) on $[0,(L-1)t]$ due to (\ref{eq:IAtx}). Hence, 
it is not hard to show that $I_A(t,x)$ is right-continuous at the left-most 
point $x=0$. Therefore, the limit $\lim_{y \rightarrow x^+} I_A(t,y)$ 
exists and we have $I_A(t,x)=I^+_A(t,x)$ for any $x \in [0,(L-1)t)$.
\end{proof}

First, we focus on the case where $L>1$, and consider three types of 
events, $\mE_1$, $\mE^c_2$, and $\mE^c_3$, that imply the delay-violation 
event $\{ W(0)>b \}$. 

{\bf Event} $\mE_1$: Suppose that there is a packet that arrives to 
the network in time-slot $-b-1$. Without loss of generality, we assume
that the packet arrives to queue $Q_1$. Further, suppose that $Q_1$ is 
disconnected from all $n$ servers in all time-slots from $-b-1$ to $-1$. 

Then, at the beginning of time-slot 0, this packet is still in the 
network and has a delay of $b+1$. This implies $\mE_1 \subseteq \{W(0)>b\}$. 
Note that the probability that event $\mE_1$ occurs can be computed as
\[
\Prob(\mE_1) = (1-q)^{n(b+1)}=e^{-n(b+1)I_X}.
\]
Hence, we have 
\[
\Prob(W(0)>b) \ge e^{-n(b+1)I_X},
\]
and thus
\[
\limsup_{n \rightarrow \infty} \frac {-1}{n} \log \Prob(W(0)>b) \le (b+1)I_X.
\]

{\bf Event} $\mE^c_2$: Consider any fixed $c \in \{0,1,\dots,b\}$ and any 
$t > t_{b-c}$. Recall that $t_{b-c} = \frac {b-c}{L-1}$. Then, for all 
$t > t_{b-c}$, we have $b-c < (L-1)t$, and thus $I_A(t,b-c) = I^+_A(t,b-c)$ 
from Lemma~\ref{lem:IAtx}. Hence, for any fixed $\epsilon>0$, there exists 
a $\delta>0$ such that $I_A(t,b-c+\delta) \le I^+_A(t,b-c)+\epsilon = I_A(t,b-c)
+\epsilon$. Suppose that from time-slot $-t-b$ to $-b-1$, the total number 
of packet arrivals to the system is greater than or equal to $nt + n(b-c+\delta)$, 
and let $p_{(b-c+\delta)}$ denote the probability that this event occurs. 
Then, from Cramer's Theorem, we have $\lim_{n \rightarrow \infty} \frac 
{-1}{n} \log p_{(b-c+\delta)}=I_A(t,b-c+\delta) \le I_A(t,b-c)+\epsilon$.
Clearly, the total number of packets that are served in any time-slot is 
no greater than $n$. For any fixed $\delta$, we have $n\delta \ge 1$ for 
large enough $n$ (when $n \ge \frac {1}{\delta}$). Hence, if the above 
event occurs, at the end of 
time-slot $-c-1$, the system contains at least one packet that arrived before 
time-slot $-b$. Without loss of generality, we assume that this packet is in 
$Q_1$. Now, assume that $Q_1$ is disconnected from all $n$ servers in the 
next $c$ time-slots, i.e., from time-slot $-c$ to $-1$. This occurs with 
probability $(1-q)^{cn}=e^{-ncI_X}$, independently of all the past history. 
Hence, at the beginning of time-slot 0, there is still a packet that arrived 
before time-slot $-b$. Thus, we have $W(0)>b$ in this case. This implies 
$\mE^c_2 \subseteq \{W(0)>b\}$. Note that the probability that event $\mE^c_2$ 
occurs can be computed as
\[
\Prob(\mE^c_2) = p_{(b-c+\delta)}e^{-ncI_X}.
\]
Hence, we have 
\[
\Prob(W(0)>b) \ge p_{(b-c+\delta)}e^{-ncI_X},
\]
and thus
\[
\limsup_{n \rightarrow \infty} \frac {-1}{n} \log \Prob(W(0)>b) \le I_A(t,b-c) + \epsilon + cI_X.
\]

Since the above inequality holds for any $c \in \{0,1,\dots,b\}$, 
any $t>t_{b-c}$, and any $\epsilon>0$, by letting $\epsilon$ tend 
to 0, taking the infimum over all $t>t_{b-c}$, and taking the minimum 
over all $c \in \{0,1,\dots,b\}$, we have
\[
\begin{split}
\limsup_{n \rightarrow \infty} & \frac {-1}{n} \log \Prob(W(0)>b) \\
& \le \min_{c \in \{0,1,\dots,b \}} \{ \inf_{t>t_{b-c}} I_A(t,b-c) + cI_X \}.
%& = \min_{c \in \{0,1,\dots,b \}} \{ \inf_{t>0, t \neq t_{b-c}} I^+_A(t,b-c) + cI_X \}, \\
\end{split}
\]
%where the equality is from $I^+_A(t,b-c) = \infty$ for $t<t_{b-c}$.

{\bf Event} $\mE^c_3$: Consider any fixed $c \in \Psi_b$.
Suppose that from time-slot $-t_{b-c}-b$ to $-b-1$, the total number of packet 
arrivals to the system is equal to $nt_{b-c} + n(b-c) = nLt_{b-c}$, and let 
$p^{\prime}_{(b-c)}$ denote the probability that this event occurs. Note that
the total number of packet arrivals to the system from time-slot $-t_{b-c}-b$ 
to $-b-1$ can never exceed $nLt_{b-c}$. Then, from Cramer's Theorem, we have 
$\lim_{n \rightarrow \infty} \frac {-1}{n} \log p^{\prime}_{(b-c)}=I_A(t_{b-c},b-c)$. 
Clearly, the total number of packets that can be served during the interval 
$[-t_{b-c}-b, -c-1]$ is no greater than $n(t_{b-c}+b-c) = nLt_{b-c}$. Suppose 
that there exists one queue that is disconnected from all the servers in any 
one time-slot in the interval $[-t_{b-c}-b, -c-1]$. Then, at the end of time-slot 
$-c-1$, the system contains at least one packet that 
arrived before time-slot $-b$. Without loss of generality, we assume that queue 
$Q_1$ is disconnected from all the servers in a time-slot, say time-slot $-t_{b-c}-b$. 
This event occurs with probability $(1-q)^n=e^{-nI_X}$. Further, assume that $Q_1$ 
is disconnected from all the $n$ servers in the next $c$ time-slots, i.e., from 
time-slot $-c$ to $-1$. This occurs with probability $(1-q)^{cn}=e^{-ncI_X}$, 
independently of all the past history. Hence, at the beginning of time-slot 0, 
there is still a packet that arrived before time-slot $-b$. Thus, we have $W(0)>b$ 
in this case. This implies $\mE^c_3 \subseteq \{W(0)>b\}$. Note that the probability 
that event $\mE^c_3$ occurs can be computed as
\[
\Prob(\mE^c_3) = p^{\prime}_{(b-c)}e^{-n(c+1)I_X}.
\]
Hence, we have 
\[
\Prob(W(0)>b) \ge p^{\prime}_{(b-c)}e^{-n(c+1)I_X},
\]
and thus
\[
\limsup_{n \rightarrow \infty} \frac {-1}{n} \log \Prob(W(0)>b) \le I_A(t_{b-c},b-c) + (c+1)I_X.
\]
Since the above inequality holds for any $c \in \Psi_b$, by taking 
the minimum over all $c \in \Psi_b$, we have, for $L > 1$,
\[
\begin{split}
\limsup_{n \rightarrow \infty} & \frac {-1}{n} \log \Prob(W(0)>b) \\
& \le \min_{c \in \Psi_b} \{ I_A(t_{b-c},b-c) + (c+1)I_X \}.
\end{split}
\]

Considering events $\mE_1$, $\mE^c_2$, and $\mE^c_3$, we have
\[
\begin{split}
\limsup_{n \rightarrow \infty} & \frac {-1}{n} \log \Prob(W(0)>b) \\ 
& \le \min \{(b+1)I_X, \\
&~~~~~~~~~~ \min_{0 \le c \le b} \{\inf_{t>t_{b-c}} I_A(t,b-c) + c I_X \}, \\
&~~~~~~~~~~ \min_{c \in \Psi_b} \{ I_A(t_{b-c},b-c) + (c+1) I_X \} \} \\
& = I_0(b).
\end{split}
\]

Next, we consider the case where $L=1$. In this case, we only need to
consider event $\mE_1$, and we have $\limsup_{n \rightarrow \infty} 
\frac {-1}{n} \log \Prob(W(0)>b) \le (b+1) I_X$. 

Combining both cases of $L=1$ and $L>1$, we have $\limsup_{n \rightarrow 
\infty} \frac {-1}{n} \log \Prob(W(0)>b) \le I^*_0(b)$. This completes 
our proof.

%%%%%%%%%%%%%%%%%%%%%%%%%%%%%%%%%%%%%%%%%%%%%%%%%%%%%%%%%%%%%%%%%%%%%
\section{Proof of Theorem~\ref{thm:suff-rf}}  \label{app:thm:suff-rf}
%%%%%%%%%%%%%%%%%%%%%%%%%%%%%%%%%%%%%%%%%%%%%%%%%%%%%%%%%%%%%%%%%%%%%
Suppose policy $\mathbf{P}$ satisfies the sufficient condition in 
Theorem~\ref{thm:suff-rf}. We want to show that for any given integer 
threshold $b \ge 0$, the rate-function attained by policy $\mathbf{P}$ 
is no smaller than $I^*_0(b)$. The proof follows a similar argument as in 
the proof of Theorem~2 in \cite{sharma11b}. However, our proof exhibits 
the following key difference. In \cite{sharma11b}, the authors prove 
that the FBS policy can attain a certain rate-function, which, in some 
cases only, meets the upper bound derived in \cite{sharma11b}. In contrast, 
in the following proof, by exploiting the dominance property over both the 
FBS policy and the perfect-matching policy in Lemma~\ref{lem:dom}, we will 
show that the rate-function attained by policy $\mathbf{P}$ is always no 
smaller than the upper bound $I^*_0(b)$ that we derived in Theorem~\ref{thm:ub} 
and is thus optimal.

We first consider the case of $L>1$, and want to show that the rate-function 
attained by policy $\mathbf{P}$ is no smaller than $I_0(b)$.

In the following proof, we will use the dominance property of policy $\mathbf{P}$ 
over the FBS policy and the perfect-matching policy considered in Lemma~\ref{lem:dom}. 
We first choose the value of parameter $h$ for FBS based on the statistics of the 
arrival process. We fix $\delta<\frac{2}{3}$ and $\epsilon<\frac{p}{2}$. Then, from 
Assumption~\ref{ass:arr_ld}, there exists a positive function $I_B(\epsilon,\delta)$ 
such that for all $n \ge N_B(\epsilon,\delta)$ and $t \ge T_B(\epsilon,\delta)$, we have
\[
\Prob \left( \frac {\sum_{\tau=l+1}^{l+t} \mathbb{1}_{\{|A(\tau) - pn|
> \epsilon n \}}} {t} > \delta \right) < \exp (-nt I_B(\epsilon,\delta)),
\]
for any integer $l$. We then choose 
\[
h = \max \left \{ T_B(\epsilon,\delta), 
\left \lceil \frac {1}{(p-\epsilon)(1-\frac{3\delta}{2})} \right \rceil,  
\left \lceil \frac {2I_0(b)}{I_B(\epsilon,\delta)} \right \rceil \right \} + 1.
\]
The reason for choosing the above value of $h$ will become clear later on.
Recall from Assumption~\ref{ass:arr_bound} that $L$ is the maximum number
of packets that can arrive to a queue in any time-slot $t$. Let $H=Lh$. Then, 
$H$ is the maximum number of packets that can arrive to a queue during an 
interval of $h$ time-slots, and is thus the maximum number of packets from 
the same queue in a frame.

Let $L(-b)$ be the last time before time-slot $-b$, when the backlog is 
empty, i.e., all the queues have a queue-length of zero. Also, let $\mE_t$ 
be the set of sample paths such that $L(-b)=-t-b-1$ and $W(0)>b$ under 
policy $\mathbf{P}$. 
Then, we have 
\begin{equation}
\label{eq:overflow}
\Prob (W(0)>b) = \sum_{t=1}^{\infty} \Prob(\mE_t).
\end{equation}

Let $\mE^F_t$ and $\mE^{PM}_t$ be the set of sample paths such that 
given $L(-b)=-t-b-1$, the event $W(0)>b$ occurs under the FBS policy 
and the perfect-matching policy, respectively. Recall that policy 
$\mathbf{P}$ dominates both the FBS policy and the perfect-matching 
policy. Then, for any $t>0$ we have
\begin{equation}
\label{eq:dom}
\mE_t \subseteq \mE^F_t \cap \mE^{PM}_t.
\end{equation}

Recall that $p$ is the mean arrival rate to a queue. Now, we choose any 
fixed real number $\hat{p} \in (p,1)$, and fix a finite time $t^*$ as 
\begin{equation}
\label{eq:tstar}
t^* \triangleq \max \{ T_1, \left \lceil \frac 
{I_0(b)} { I_{BX} } \right \rceil, \max \{t_{b-c} ~|~ c \in \Psi_b\} \}, 
\end{equation}
where 
\begin{equation}
\label{eq:t1}
T_1 \triangleq \max \{
T_B(\hat{p}-p, \frac {1-\hat{p}} {6(L+2)}), \left \lceil \frac {6} {1-\hat{p}} \right \rceil \}
\end{equation}
and 
\begin{equation}
\label{eq:ibx}
I_{BX} \triangleq \min \{ \frac {(1-\hat{p})I_X}{9}, I_B(\hat{p}-p,\frac {1-\hat{p}} {6(L+2)}) \}.
\end{equation}
The reason for defining the above value of $t^*$ will become clear later on. 
Then, we apply (\ref{eq:dom}) to (\ref{eq:overflow}) and split the summation as
\[
\Prob (W(0)>b) \le P_1 + P_2,
\]
where
\[
P_1 \triangleq \sum_{t=1}^{t^*} \Prob (\mE^F_t \cap \mE^{PM}_t)
\]
and 
\[
P_2 \triangleq \sum_{t=t^*}^{\infty} \Prob (\mE^F_t \cap \mE^{PM}_t).
\]

We divide the proof into two parts. In Part 1, we show that there exists 
a finite $N_1>0$ such that for all $n \ge N_1$, we have
\[
P_1 \le C_1 n^{7bH} e^{-n I_0(b)}.
\]
Then, in Part 2, we show that there exists a finite $N_2>0$ such that for 
all $n \ge N_2$, we have
\[
P_2 \le 4 e^{-n I_0(b)}.
\]
Finally, combining both parts, we have
\[
\Prob (W(0)>b) \le \left( C_1 n^{7bH} + 4 \right) e^{-n I_0(b)},
\]
for all $n \ge N \triangleq \max \{N_1, N_2\}$. By taking logarithm and 
limit as $n$ goes to infinity, we obtain $\liminf_{n \rightarrow \infty} 
\frac {-1} {n} \log \Prob \left( W(0) > b \right) \ge I_0(b)$, and thus 
the desired results.

%The detailed proof is provided in our online technical report \cite{ji13a}.

Before we prove Part 1 and Part 2, we derive the following properties of
the FBS policy and the perfect-matching policy, which will be used in the proof.

We first calculate an upper bound on the probability that during interval 
$[-t-b,-1]$, there are exactly $t+a$ frames that can be served by the FBS 
policy, for some $a \le b$. We define the random variable $X_F(t)=1$ if 
there exists a schedule such that a frame can be successfully served in 
time-slot $t$ under FBS, and $X_F(t)=0$, otherwise. Let $X_F(i,j)$ denote 
the total number of frames served by FBS from time-slot $i$ to $j$. Given
any sample path in the set $\mE_t$, i.e., $L(-b)=-t-b-1$ and $W(0)>b$, the 
backlog never becomes empty under policy $\mathbf{P}$ during the interval 
of $[-t-b,-1]$. Since policy $\mathbf{P}$ dominates the FBS policy, then 
the buffer is non-empty under the FBS policy during the interval of $[-t-b,-1]$. 
Hence, we have
\begin{equation}
\label{eq:xf}
X_F(-t-b,-1) = \sum_{\tau=-t-b}^{-1} X_F(\tau).
\end{equation}
From Lemma~6 of \cite{sharma11b}, there exists an $N_F >0$, such that for 
all $n \ge N_F$ the probability that there does not exist a schedule such 
that a frame can be served in each time-slot is no greater than 
$(\frac {n} {1-q})^{7H} e^{-n\log \frac {1}{1-q}}$. Hence, we have
\begin{equation}
\label{eq:ubs}
\begin{split}
\Prob ( & \sum_{\tau=-t-b}^{-1} X_F(\tau) = t+a) \\
& \le \binom{t+b} {t+a} \left(\frac {n} {1-q}\right)^{7bH} e^{-(b-a)n \log \frac {1}{1-q}} \\
& \le 2^{t+b} \left(\frac {n} {1-q}\right)^{7bH} e^{-n(b-a) I_X}.
\end{split}
\end{equation}

Next, we define the random variable $X_{PM}(t)=1$ if a perfect-matching can be found 
in time-slot $t$, and $X_{PM}(t)=0$, otherwise. From Lemma~1 of \cite{bodas09}, there 
exists an $N_{PM} >0$, such that for all $n \ge N_{PM}$ the probability that no perfect
matching can be found in each time-slot is no greater than $3n e^{-n\log \frac {1}{1-q}}$.
Then, we can similarly show that 
\begin{equation}
\label{eq:ubs_pm}
\Prob (\sum_{\tau=-t-b}^{-1} X_{PM}(\tau) = t+a) \le 2^{t+3b} n^b e^{-n(b-a) I_X}.
\end{equation}
It is easy to observe that the right hand side of (\ref{eq:ubs}) and (\ref{eq:ubs_pm}) 
is a monotonically increasing function in $a$.

{\bf Part 1:} Consider any $t \in \{1,2,\dots,t^*\}$. 
We let $\mE^\alpha_t$ denote the set of sample paths in which there are at 
least $n$ packet arrivals to the system during every $h-1$ time-slots in the 
interval of $[-t-b,-b-1]$. Let $\mE^\beta_t$ denote the set of sample paths 
in which $\frac {A(-t-b,-b-1)}{n_0} - \sum_{\tau=-t-b}^{-1} X_F(\tau)>0$
under the FBS policy. 
Since in any sample path of $\mE_t$ the backlog never becomes empty (under 
policy $\mathbf{P}$) during the interval of $[-t-b,-b-1]$  and policy $\mathbf{P}$ 
dominates the FBS policy, then in any sample path of $\mE^\alpha_t$ the backlog 
never becomes empty (under the FBS policy) during the interval of $[-t-b,-b-1]$.
Similar as in the proof of Theorem~2 of \cite{sharma11b}, using Lemma~9 of
\cite{sharma11b}, we can show that
\begin{equation}
\label{eq:pfab}
\mE^F_t \subseteq (\mE^\alpha_t)^c \cup \mE^\beta_t,
\end{equation}
and along with the choice of $h$ (as chosen earlier), we can show that 
there exist $N_3>0$ and $C_2>0$ such that for all $n \ge N_3$, 
\begin{equation}
\label{eq:palpha}
\Prob (\mE^\alpha_t) > 1 - C_2 t e^{-n I_0(b)}.
\end{equation}
Here, we do not duplicate the detailed proofs for (\ref{eq:pfab}) and 
(\ref{eq:palpha}), and refer the interested readers to \cite{sharma11b} 
for details.

Next, we compute the probability of $\mE^\beta_t$ for each $t$. For any fixed 
integer $t>0$, we derive an upper bound for the probability of a large burst 
of arrivals during an interval of $t$ time-slots. Let $\theta_t \triangleq 
\argmax_\theta [\theta(t+x)-\lambda_{A_i(-t+1,0)}(\theta)]$, and let $\theta^* 
\triangleq \max\{\theta_1,\theta_2,\dots,\theta_{t^*}\}$. Recall that $n_0=n-H$. 
We know from the Chernoff bound that for any $x \in [0,(L-1)t]$,
\begin{equation}
\label{eq:uba}
\begin{split}
& \Prob (A(-t+1,0) > n_0(t+x)) \\
& = \Prob (A(-t+1,0) \ge (n-H)(t+x)+1) \\
& \le e^{-n(\theta_t(t+x)-\lambda_{A_i(-t+1,0)}(\theta_t))+(H(t+x)-1)\theta_t} \\
& \le e^{-nI_A(t,x)} e^{(H(t+x)-1)\theta^*}.
\end{split}
\end{equation}

Recall that $t_x = \frac {x}{L-1}$. We first consider any $t \in \{1,2,\dots,t^*\} 
\backslash \{t_{b-c^{\prime}} ~|~ c^{\prime} \in \Psi_b \}$. For these values of
$t$, we will use the dominance property over FBS, i.e., $\mE_t \subseteq \mE^F_t$.
Let $c_t$ be such that $t_{b-c_t} < t < t_{b-c_t+1}$, or $c_t=0$ if $t > t_b$. Then, 
for all $z \in \{c_t,c_t+1,\dots,b\}$, we have $t_{b-z} < t$, and thus $t+b-z < Lt$;
and for all $z^{\prime} < c_t$, we have $t_{b-z^{\prime}} > t$, and thus $t+b-z^{\prime} 
> Lt$. Using the results from (\ref{eq:xf}), (\ref{eq:ubs}) and (\ref{eq:uba}), we 
have that for all $n \ge N_F$,
\[
\begin{split}
& \Prob (\mE^\beta_t) \\
& = \Prob ( \frac {A(-t-b,-b-1)} {n_0} - X_F(-t-b,-1) > 0 ) \\
& = \sum_{a=0}^{t+b-c_t} \Prob ( \sum_{\tau=-t-b}^{-1} X_F(\tau) = a ) \Prob ( A(-t-b,-b-1) > an_0)  \\
& \le (t+b+1) \max_{0 \le a \le t+b-c_t} \{ \Prob ( \sum_{\tau=-t-b}^{-1} X_F(\tau) = a ) \\
&~~~ \times \Prob ( A(-t-b,-b-1) > an_0) \} \\
& \le (t+b+1) \max \{ \max_{a \in \{0,1, \dots, t-1 \}} \{\Prob ( \sum_{\tau=-t-b}^{-1} X_F(\tau) = a )\}, \\
&~~~ \max_{a \in \{0, \dots, b-c_t \}} \{ \Prob ( \sum_{\tau=-t-b}^{-1} X_F(\tau) = t+a ) \\
&~~~ \times \Prob ( A(-t-b,-b-1) > (t+a)n_0) \} \} \\
& \stackrel{(a)}\le (t+b+1) \max \{ 2^{t+b} \left(\frac {n} {1-q} \right)^{7bH} e^{-n(b+1) I_X}, \\
&~~~ \max_{a \in \{0, \dots, b-c_t \}} \{ \Prob ( \sum_{\tau=-t-b}^{-1} X_F(\tau) = t+a ) \\
&~~~ \times \Prob ( A(-t-b,-b-1) > (t+a)n_0) \} \} \\
& \stackrel{(b)}\le (t+b+1) 2^{t+b} \left(\frac {n} {1-q} \right)^{7bH} e^{(H(t+b)-1)\theta^*} \\
&~~~ \times \max \{ e^{-n(b+1) I_X}, \max_{a \in \{0, \dots, b-c_t \}} \{ e^{-n(I_A(t,a) + (b-a) I_X)} \} \} \\
& \stackrel{(c)}\le C_3 n^{7bH} e^{-n \min \{ (b+1) I_X, \min_{z \in \{c_t,c_t+1,\dots,b \}} \{I_A(t,b-z) + z I_X\} \} },
\end{split}
\]
where $C_3 \triangleq (t^*+b+1) 2^{t^*+b} (\frac {1}{1-q})^{7bH} e^{(H(t^*+b)-1)\theta^*}$, 
(a) is from the monotonicity of the right hand side of (\ref{eq:ubs}),
(b) is from (\ref{eq:ubs}) and (\ref{eq:uba}), and (c) is from changing 
variable by setting $z=b-a$.

Recall that
\[
\begin{split}
I_0(b) \triangleq \min \{ &(b+1)I_X, \\
& \min_{0 \le c \le b} \{\inf_{t>t_{b-c}} I_A(t,b-c) + c I_X \}, \\
& \min_{c \in \Psi_b} \{ I_A(t_{b-c},b-c) + (c+1) I_X \} \}.
\end{split}
\]
Hence, for all $n \ge N_4 \triangleq \max \{N_3, N_F \}$, we have
\[
\begin{split}
\Prob (\mE^F_t \cap \mE^{PM}_t) & \le \Prob (\mE^F_t) \\
& \stackrel{(a)} \le 1-\Prob (\mE^\alpha_t) + \Prob (\mE^\beta_t) \\
& \stackrel{(b)} \le C_4 n^{7bH} e^{-n I_0(b)},
\end{split}
\]
where $C_4 \triangleq \max \{ C_2 t^*, C_3 \}$, (a) is from (\ref{eq:pfab}), 
and (b) is from (\ref{eq:palpha}) and the above result.

Next, we consider any $t_{b-c} \in \{t_{b-c^{\prime}} ~|~ c^{\prime} \in \Psi_b\}$. 
For these values of $t$, we will use the dominance property over both the FBS policy
and the perfect-matching policy, i.e., $\mE_t \subseteq \mE^F_t$ and $\mE_t \subseteq 
\mE^{PM}_t$. Recall that we have $t_{b-c} = \frac {b-c} {L-1}>0$, and thus 
\begin{equation}
\label{eq:tbc}
t_{b-c} + b-c = L t_{b-c}.
\end{equation}
We first split $\Prob (\mE^F_{t_{b-c}} \cap \mE^{PM}_{t_{b-c}})$ as 
\begin{equation}
\label{eq:k12}
\Prob (\mE^F_{t_{b-c}} \cap \mE^{PM}_{t_{b-c}}) \le K_1 + K_2,
\end{equation}
where 
\[
\begin{split}
K_1 &\triangleq \Prob(\mE^F_{t_{b-c}} \cap \mE^{PM}_{t_{b-c}}, \\
&~~~~~~ A(-t_{b-c}-b,-b-1) < (t_{b-c}+b-c)n_0) \\
& \le \Prob (\mE^F_{t_{b-c}}, A(-t_{b-c}-b,-b-1) < (t_{b-c}+b-c)n_0) \\
& \le 1-\Prob (\mE^\alpha_{t_{b-c}}) + K^{\prime}_1,
\end{split}
\]
\[
K^{\prime}_1 \triangleq \Prob(\mE^\beta_{t_{b-c}}, A(-t_{b-c}-b,-b-1) < (t_{b-c}+b-c)n_0),
\]
and
\[
\begin{split}
K_2 & \triangleq \Prob(\mE^F_{t_{b-c}} \cap \mE^{PM}_{t_{b-c}}, \\
&~~~~~~ A(-t_{b-c}-b,-b-1) \ge (t_{b-c}+b-c)n_0) \\
& \le \Prob(\mE^{PM}_{t_{b-c}}, A(-t_{b-c}-b,-b-1) \ge (t_{b-c}+b-c)n_0).
\end{split}
\]
In the above derivation, we use the dominance over the FBS policy and the 
perfect-matching policy for $K_1$ and $K_2$, respectively.

Similarly, using the results from (\ref{eq:xf}), (\ref{eq:ubs}), and (\ref{eq:uba}),
we have that for all $n \ge N_F$,
\[
\begin{split}
K^{\prime}_1
&= \Prob ( \frac {A(-t_{b-c}-b,-b-1)} {n_0} - \sum_{\tau=-t_{b-c}-b}^{-1} X_F(\tau) > 0, \\
&~~~~~~ A(-t_{b-c}-b,-b-1) < (t_{b-c}+b-c)n_0) \\
& \le \sum_{a=0}^{t_{b-c}+b-c-1} (\Prob ( \sum_{\tau=-t_{b-c}-b}^{-1} X_F(\tau) = a ) \\
&~~~~~~~~~~~~~ \times \Prob ( A(-t_{b-c}-b,-b-1) > an_0))  \\
& \le (t_{b-c}+b-c) \max_{0 \le a \le t_{b-c}+b-c-1} \{ \Prob ( \sum_{\tau=-t_{b-c}-b}^{-1} X_F(\tau) = a ) \\
&~~~~~ \times \Prob ( A(-t_{b-c}-b,-b-1) > an_0) \} \\
% & \le (t_{b-c}+b-c) \max \{ \max_{a \in \{0,1, \dots, t_{b-c}-1 \}} \{\Prob ( \sum_{\tau=-t_{b-c}-b}^{-1} X_F(\tau) = a )\}, \\
% &~~~~~ \max_{a \in \{0, \dots, b-c-1 \}} \{ \Prob ( \sum_{\tau=-t_{b-c}-b}^{-1} X_F(\tau) = t_{b-c}+a ) \\
% &~~~~~ \times \Prob ( A(-t_{b-c}-b,-b-1) > (t_{b-c}+a)n_0) \} \} \\
% & \le (t_{b-c}+b+1) \max \{ 2^{t_{b-c}+b} \left(\frac {n} {1-q} \right)^{7bH} e^{-n(b+1) I_X}, \\
% &~~~~~ \max_{a \in \{0, \dots, b-c-1 \}} \{ \Prob ( \sum_{\tau=-t_{b-c}-b}^{-1} X_F(\tau) = t_{b-c}+a ) \\
% &~~~~~ \times \Prob ( A(-t_{b-c}-b,-b-1) > (t_{b-c}+a)n_0) \} \} \\
% & \le (t_{b-c}+b+1) 2^{t_{b-c}+b} \left(\frac {n} {1-q} \right)^{7bH} e^{(H(t_{b-c}+b)-1)\theta^*} \\
% &~~~~~ \times \max \{ e^{-n(b+1) I_X}, \max_{a \in \{0, \dots, b-c-1 \}} \{ e^{-n(I_A(t_{b-c},a) + (b-a)I_X)}  \} \} \\
& \le C_3 n^{7bH} e^{-n \min \{ (b+1) I_X, \min_{z \in \{c+1,\dots,b \}} \{I_A(t_{b-c},b-z) + z I_X\} \} },
\end{split}
\]
and using the results from (\ref{eq:ubs_pm}) and (\ref{eq:uba}),
we have that for all $n \ge N_{PM}$, 
\[
\begin{split}
K_2
&\le \Prob(\mE^{PM}_{t_{b-c}} ~|~ A(-t_{b-c}-b,-b-1) \ge (t_{b-c}+b-c)n_0) \\
&~~~ \times \Prob (A(-t_{b-c}-b,-b-1) \ge (t_{b-c}+b-c)n_0) \\
& \stackrel{(a)}\le \Prob(\mE^{PM}_{t_{b-c}} ~|~ A(-t_{b-c}-b,-b-1) = Lnt_{b-c}) \\
&~~~ \times e^{-nI_A(t_{b-c},b-c)} e^{(H(t_{b-c}+b-c)-1)\theta^*} \\
& \stackrel{(b)}\le \Prob(\sum_{\tau=-t_{b-c}-b}^{-1} X_{PM}(\tau) < t_{b-c}+b-c) \\
&~~~ \times e^{-nI_A(t_{b-c},b-c)} e^{(H(t_{b-c}+b-c)-1)\theta^*} \\
&\le \sum_{a=0}^{t_{b-c}+b-c-1} \Prob(\sum_{\tau=-t_{b-c}-b}^{-1} X_{PM}(\tau) = a) \\
&~~~ \times e^{-nI_A(t_{b-c},b-c)} e^{(H(t_{b-c}+b-c)-1)\theta^*} \\
&\le (t_{b-c}+b-c) \max_{a \in \{0,...,t_{b-c}+b-c-1\}} \Prob(\sum_{\tau=-t_{b-c}-b}^{-1} X_{PM}(\tau) = a) \\
&~~~ \times e^{-nI_A(t_{b-c},b-c)} e^{(H(t_{b-c}+b-c)-1)\theta^*} \\
& \stackrel{(c)}\le (t_{b-c}+b-c) 2^{t+3b} e^{(H(t_{b-c}+b-c)-1)\theta^*} n^b \\
&~~~ \times e^{-n(I_A(t_{b-c},b-c) + (c+1) I_X)} \\
&\le C_5 n^b e^{-n(I_A(t_{b-c},b-c) + (c+1) I_X)},
\end{split}
\]
where $C_5 \triangleq (t^*+b) 2^{t^*+3b} e^{(H(t^*+b)-1)\theta^*}$,
(a) is from (\ref{eq:uba}), (b) is because $t_{b-c}= \frac {b-c}{L-1}$
and the perfect-matching policy serves exactly $n$ packets in each time-slot
when $A(-t_{b-c}-b,-b-1) = Lnt_{b-c}$, and (c) is from (\ref{eq:ubs_pm})
and the monotonicity of its right hand side.

Hence, from the above results, (\ref{eq:palpha}), and (\ref{eq:k12}), 
we have that for any $t_{b-c}$ with $c \in \Psi_b$, and for all 
$n \ge N_5 \triangleq \max \{ N_3,N_F,N_{PM}\}$,
\[
\Prob (\mE^F_{t_{b-c}} \cap \mE^{PM}_{t_{b-c}}) 
\le C_6 n^{7bH} e^{-n I_0(b)},
\]
where $C_6 \triangleq \max \{C_2 t^*, C_3, C_5\}$.

Summing over $t=1$ to $t^*$, we have 
\[
\begin{split}
P_1 &= \sum_{t=1}^{t^*} \Prob (L(-b)=-t_{b-c}-b-1, \mE_t) \\
&\le C_1 n^{7bH} e^{-n I_0(b)},
\end{split}
\]
for all $n \ge N_1 \triangleq \max \{N_4, N_5 \}$, where $C_1 \triangleq C_6 t^*$.

{\bf Part 2:}
We want to show that there exists an $N_2>0$ such that for all $n \ge N_2$, 
we have
\[ 
P_2 \le 4 e^{-n I_0(b)}.
\]

Before we proceed, we first provide an informal discussion on the intuition 
behind. In Part 1, as we have seen, the delay-violation event could occur 
due to both bursty arrivals and sluggish services. However, when time interval 
$t$ is large enough, from Assumption~\ref{ass:arr_ld}, we know that the total
arrivals to the system will not deviate far away from its mean $npt$ during
an interval of $t$ time-slots for large $n$. On the other hand, if FBS can 
find a schedule to serve the HOL frame in each time-slot during an interval 
of $t$ time-slots, the total service can sum up to $(n-Lh)t$. Hence, the 
delay-violation event occurs mostly due to sluggish services when $t$ is large.

Let $R_0$ be the empty space in the end-of-line frame at the end of time-slot 
$t_1$. Then, let $A^{R_0}_F(t_1,t_2)$ denote the number of new frames created 
from time-slot $t_1$ to $t_2$, including any partially-filled frame in time-slot 
$t_2$, but excluding the partially-filled frame in time-slot $t_1$. Also, 
let $A_F(t_1,t_2)=A^{R_0}_F(t_1,t_2)$, if $R_0=0$. As in the proof for Theorem~2 
of \cite{sharma11b}, for any fixed real number $\hat{p} \in (p,1)$, we consider 
the arrival process $\hat{A}(\cdot)$, by adding extra dummy arrivals to the 
original arrival process $A(\cdot)$. The resulting arrival process $\hat{A}(\cdot)$ 
is simple, and has the following property:
\[
\hat{A}(\tau) = \left\{
\begin{array}{ll}
\hat{p} n, & \text{if}~ A(\tau) \le \hat{p} n,\\
L n, & \text{if}~ A(\tau) > \hat{p} n.
\end{array}
\right. 
\]
Hence, if we can find an upper bound on $\hat{A}_F(-t-b,-b-1)$, by our 
construction, then it is also an upper bound on $A_F(-t-b,-b-1)$.

Consider any $t \ge t^*$. Let $B=\{b_1, b_2, \dots, b_{|B|}\}$ be the set 
of time-slots in the interval from $-t-b$ to $-b-1$ when $\hat{A}(\tau) = Ln$. 
Given $L(-b)=-t-b-1$, from Corollary~2 of \cite{sharma11}, we have that,
\[
\begin{split}
\hat{A}_F( & -t-b,-b-1) \\
& \le \sum_{r=1}^{|B|-1} \left \lceil \frac {\hat{A}(b_r + 1, b_{r+1}-1)} {n_0} 
\right \rceil + \sum_{r=1}^{|B|} \left \lceil \frac {\hat{A}(b_r,b_r)} {n_0} \right \rceil \\
&~~~~~ + \left \lceil \frac {\hat{A}(-t-b,b_1 - 1)} {n_0} \right \rceil
+ \left \lceil \frac {\hat{A}(b_{|B|} + 1, -b-1)} {n_0} \right \rceil \\
& \le \sum_{r=1}^{|B|-1} \frac {\hat{A}(b_{r+1} - 1 - b_r)\hat{p}n} {n_0} + |B|-1
+ \sum_{r=1}^{|B|} \frac {Ln} {n_0} + |B| \\
&~~~~~ + \frac {(b_1+t+b) \hat{p} n} {n_0} + 1 + \frac {(-b-1 - b_{|B|})\hat{p} n} {n_0} + 1 \\
& \le \frac {(t-|B|)\hat{p} n} {n_0} + |B| \frac {Ln} {n_0} + 2 |B| + 1 \\
& \le \frac {n} {n_0} (\hat{p}t + (L+2)|B| + 1).
\end{split}
\]

From Assumption~\ref{ass:arr_ld} on the arrival process we know that for 
large enough $n$ and $t$, $|B|$ can be made less than an arbitrarily small
fraction of $t$. Further, we can show that for $n \ge \frac {(2+\hat{p})H} 
{1-\hat{p}}, t > \frac {6} {1-\hat{p}}$ and $|B|<\frac {1-\hat{p}}{6(L+2)}t$, 
we have $\hat{A}_F(-t-b,-b-1) < (\frac {2+\hat{p}}{3})t$. This is derived by 
substituting the values of $n,t$ and $|B|$ in the equation above,
\[
\begin{split}
A_F( & -t-b,-b-1) \\
& \le \hat{A}_F(-t-b,-b-1) \\
& \le \frac {n} {n_0} (\hat{p}t + (L+2)|B| + 1) \\
& < \frac {2+\hat{p}} {1+2\hat{p}} \left( \hat{p}t+\frac {1-\hat{p}} {3}t \right) \\
& \le (\frac {2+\hat{p}} {3}) t.
\end{split}
\]
Then, it follows that
\begin{equation}
\label{eq:prob_arr}
\begin{split}
& \Prob (A_F(-t-b,-b-1) \ge (\frac {2+\hat{p}} {3}) t, \\
& ~~~~~~ L(-b)=-t-b-1) \\
& = 1 - \Prob(A_F(-t-b,-b-1) < (\frac {2+\hat{p}} {3}) t, \\
& ~~~~~~~~~~~~ L(-b)=-t-b-1) \\
& \le 1 - \Prob (|B| \le \frac {1-\hat{p}} {6(L+2)} t) \\
%& = \Prob (|B| > \frac {1-\hat{p}} {6(L+2)} t) \\
& \le e^{-ntI_B(\hat{p}-p,\frac {1-\hat{p}} {6(L+2)})},
\end{split}
\end{equation}
for all $n \ge N_6 \triangleq \max \{ N_B(\hat{p}-p, \frac {1-\hat{p}} {6(L+2)}),
\frac {(2+\hat{p})H} {1-\hat{p}} \}$ and $t \ge T_1$, where the last inequality
is from Assumption~\ref{ass:arr_ld} and (\ref{eq:t1}).

Next, we state a lemma that will be used in the rest of the proof.
\begin{lemma}
\label{lem:brv}
Let $X_i$ be a sequence of binary random variables satisfying
\[
\Prob (X_i=0) < c(n) e^{-nd}, ~\text{for all}~ i,
\]
where $c(n)$ is a polynomial in $n$ of finite degree. Let $N_X$ be such 
that $c(n) < e^{\frac {nd} {2}}$ for all $n \ge N_X$. Then, for any real 
number $a \in (0,1)$, we have
\[
\Prob \left (\sum_{i=1}^{t} X_i < (1-a)t \right) \le e^{- \frac {tnad} {3}} 
\]
for all $n \ge \max \{ \frac {12} {ad}, N_X \}$
\end{lemma}

The proof follows immediately from Lemma~1 of \cite{sharma11b}.

We know from Lemma~6 of \cite{sharma11b} that for each time-slot
$\tau$, $X_F(\tau)=0$ with probability less than $(\frac {n}{1-q})^{7H} 
e^{-nI_X}$ for all $n \ge N_F$. Hence, from Lemma~\ref{lem:brv}, we have 
that there exists an $N_7 > N_F$ such that, 
\begin{equation}
\label{eq:prob_suc}
\begin{split}
& \Prob (X_F(-t-b,-1)<(\frac {2+\hat{p}} {3})t, L(-b)=-t-b-1) \\
& \le \Prob (X_F(-t-b,-1)<(\frac {2+\hat{p}} {3}) (t+b), \\
& ~~~~~~~~~ L(-b)=-t-b-1) \\
& \le e^{-n(t+b)(\frac {1-\hat{p}}{9})I_X} \\
& \le e^{-nt\frac {(1-\hat{p})I_X}{9}}, 
\end{split}
\end{equation}
for all $n \ge N_7$ and $t>0$.

From (\ref{eq:prob_arr}) and (\ref{eq:prob_suc}), we have that for all $n \ge 
N_8 \triangleq \max \{ N_6, N_7 \}$ and $t \ge T_1$,
\[
\begin{split}
& \Prob (A_F(-t-b,-b-1)-X_F(-t-b,-1)>0, \\
& ~~~~~~L(-b)=-t-b-1) \\
& \le 1-(1-e^{-nt(\frac {1-\hat{p}}{9})I_X})(1-e^{-ntI_B(\hat{p}-p,\frac {1-\hat{p}} {6(L+2)})}) \\
& \le 2 e^{-ntI_{BX}},
\end{split}
\]
where the last inequality is from (\ref{eq:ibx}).

Then, summing over all $t \ge t^*$, we have that for all $n \ge N_2 \triangleq 
\max \{ N_8, \left \lceil \frac {\log 2} {I_{BX}} \right \rceil \}$,
\[
\begin{split}
P_2 & \le \sum_{t=t^*}^{\infty} \Prob(\mE^F_t) \\
& \le \sum_{t=t^*}^{\infty} \Prob(L(-b)=-t-b-1, \\
&~~~~~~~~~~~ A_F(-t-b,-b-1) > X_F(-t-b,-1)) \\
& \le \sum_{t=t^*}^{\infty} 2 e^{-ntI_{BX}} \\
& \le \frac {2 e^{-nt^*I_{BX}}} {1-e^{-nI_{BX}}} \\
& \stackrel{(a)}\le 4 e^{-n t^* I_{BX}} \\
& \stackrel{(b)}\le 4 e^{-n I_0(b)},
\end{split}
\]
where (a) is from our choice of $N_2$, and (b) is from (\ref{eq:tstar}). 

Combining both parts, we complete the proof for the case of $L>1$.

Now, we consider the case of $L=1$. We want to show that for any fixed 
integer $b \ge 0$, the rate-function attained by policy $\mathbf{P}$ is 
no smaller than $(b+1)I_X$.

Similarly, we fix a finite time $t^{\prime}$ as 
\[
t^{\prime} \triangleq \max \{ T_1, \left \lceil \frac {(b+1)I_X} { I_{BX} } \right \rceil\}.
\]
Using the dominance property over the FBS policy and the perfect-matching 
policy, we split the summation in (\ref{eq:overflow}) as
\[
\Prob (W(0)>b) \le P^{\prime}_1 + P^{\prime}_2,
\]
where 
$P^{\prime}_1 \triangleq \sum_{t=1}^{t^{\prime}} \Prob(\mE^F_t \cap \mE^{PM}_t)$,
and
$P^{\prime}_2 \triangleq \sum_{t=t^{\prime}}^{\infty} \Prob(\mE^F_t \cap \mE^{PM}_t)$.

We divide the proof into two parts. 
In Part 1, we show that there exists a finite $N^{\prime}_1>0$ such 
that for all $n \ge N^{\prime}_1$, we have
\[
P^{\prime}_1 \le C^{\prime}_1 n^b e^{-n (b+1)I_X}.
\]
Then, in Part 2, we show that there exists a finite $N^{\prime}_2>0$ 
such that for all $n \ge N^{\prime}_2$, we have
\[
P^{\prime}_2 \le 4 e^{-n (b+1)I_X}.
\]
Finally, combining both parts, we have
\[
\Prob (W(0)>b) \le \left( C^{\prime}_1 n^b + 4 \right) e^{-n (b+1) I_X},
\]
for all $n \ge N^{\prime} \triangleq \max \{N^{\prime}_1, N^{\prime}_2\}$. 
By taking logarithm and limit as $n$ goes to infinity, we obtain 
$\liminf_{n \rightarrow \infty} \frac {-1} {n} \log \Prob \left( W(0) > b 
\right) \ge (b+1)I_X$, and thus the desired results.

For Part 2, by applying the same argument as in the case of $L>1$, we can 
show that there exists a finite $N^{\prime}_2>0$ such that for all $n \ge 
N^{\prime}_2$, we have $P^{\prime}_2 \le 4 e^{-n (b+1)I_X}$. Hence, it 
remains to show $P^{\prime}_1 \le C^{\prime}_1 n^b e^{-n (b+1)I_X}$. 

Consider any $t \in \{1,2,\dots,t^{\prime}\}$. 
Note that the total number of packet arrivals to the system during the 
interval of $[-t-b,-b-1]$ can never exceed $nt$ when $L=1$, then event 
$\mE^{PM}_t$ does not occur if the total number of time-slots when a 
perfect-matching can be found during the interval of $[-t-b,-1]$ is no 
smaller than $t$, i.e., $\sum_{\tau=-t-b}^{-1} X_{PM}(\tau) \ge t$. Thus, 
we have $\mE^{PM}_t \subseteq \{ L(b)=-t-b-1, \sum_{\tau=-t-b}^{-1} 
X_{PM}(\tau)<t \}$. Then, we have
\[
\begin{split}
& \Prob (\mE^F_t \cap \mE^{PM}_t) \\
& \le \Prob(\mE^{PM}_t) \\
& \le \Prob(\sum_{\tau=-t-b}^{-1} X_{PM}(\tau) < t) \\
& \le t \max_{a \in \{0,\dots,t-1\}}\Prob(\sum_{\tau=-t-b}^{-1} X_{PM}(\tau) = a) \\
& \stackrel{(a)}\le t^{\prime} 2^{t^{\prime}+3b} n^b e^{-n(b+1) I_X} \\
& \le C^{\prime}_3 n^b e^{-n(b+1) I_X},
\end{split}
\]
where $C^{\prime}_3 \triangleq t^{\prime} 2^{t^{\prime}+3b}$, 
and (a) is from (\ref{eq:ubs_pm}) and the monotonicity of its 
right hand side.

Let $C^{\prime}_1 \triangleq t^{\prime} C^{\prime}_3$. Summing 
over all $t \in \{1,\dots,t^{\prime}\}$, we have
\[
P^{\prime}_1 \le C^{\prime}_1 n^b e^{-n(b+1) I_X}.
\]

Combining both parts, we complete the proof for the case of $L=1$. 
Then, combining both cases of $L>1$ and $L=1$, the result of the 
theorem follows.

%%%%%%%%%%%%%%%%%%%%%%%%%%%%%%%%%%%%%%%%%%%%%%%%%%%%%%%%%%%
\section{Proof of Lemma~\ref{lem:dom}}  \label{app:lem:dom}
%%%%%%%%%%%%%%%%%%%%%%%%%%%%%%%%%%%%%%%%%%%%%%%%%%%%%%%%%%%
Suppose policy $\mathbf{P}$ satisfies the sufficient condition in 
Theorem~\ref{thm:suff-rf}. We first want to show that policy 
$\mathbf{P}$ dominates the version of the FBS policy described in 
Section~\ref{subsec:rfdo}. The proof follows a similar argument as 
in the proof of Lemma~7 in \cite{sharma11b}.
 
Consider two queueing systems, $\bar{Q}_1$ and $\bar{Q}_2$, both of 
which have the same arrival and channel realizations. We assume that 
$\bar{Q}_1$ adopts policy $\mathbf{P}$ and $\bar{Q}_2$ adopts the FBS 
policy. Recall that the weight of a packet $p$ in time-slot $t$ is 
defined as $\hat{w}(p) = t - t_p + \frac {L+1-x_p} {(L+1)} + \frac 
{n+1-q(p)} {(L+1) (n+1)} $. For two packets $p_1$ and $p_2$, we say 
$p_1$ is older than $p_2$ if $\hat{w}(p_1) > \hat{w} (p_2)$. 

Let $R_i(t)$ represent the set of packets present in the system
$\bar{Q}_i$ at the end of time-slot $t$, for $i=1,2$. Then, it 
suffices to show that $R_1(t) \subseteq R_2(t)$ for all time $t$.
We let $A(t)$ denote the set of packets that arrive at time $t$.
Let $X_i(t)$ denote the set of packets that depart the system
$\bar{Q}_i$ at time t, for $i=1,2$. Hence, we have
$R_i(t+1) = (R_i(t) \cup A(t+1)) \backslash X_i(t+1),~\text{for}~ i=1,2$.

We then proceed the proof by contradiction. Suppose that $R_1(t) 
\nsubseteq R_2(t)$ for some time $t$. Without loss of generality,
we assume that $\tau$ is the first time such that $R_1(\tau) 
\nsubseteq R_2(\tau)$ occurs. Hence, there must exist a packet,
say $p$, such that $p \in R_1(\tau)$ and $p \notin R_2(\tau)$.
Because $\tau$ is the first time when such an event occurs, packet
$p$ must depart from the system $\bar{Q}_2$ in time-slot $\tau$,
i.e., $p \in X_2(\tau)$.

Let $B_i(v)$ denote the set of packets in $R_i(\tau-1) \cup A(\tau)$
with weight greater than or equal to $v$, for $i=1,2$. Clearly, we 
have $B_1(v) \subseteq B_2(v)$ for all $v$, as $R_1(\tau-1) \subseteq
R_2(\tau-1)$ by assumption. Since packet $p$ is served in the system
$\bar{Q}_2$ in time-slot $\tau$, we know from the operations of FBS
that all packets in $B_2(\hat{w}(p))$ must also be served in time-slot
$\tau$. This is because packet $p$ is part of the HOL frame in time-slot 
$\tau$ (as packet $p$ is served in time-slot $\tau$), and all packets 
with a weight greater than $\hat{w}(p)$ must be filled to the frames 
with higher priority than packet $p$ and thus should also belong to 
the HOL frame in time-slot $\tau$. This further implies that in the
system $\bar{Q}_1$, there exists a feasible schedule that can match 
all packets in $B_1(\hat{w}(p))$, since $B_1(\hat{w}(p)) \subseteq 
B_2(\hat{w}(p))$ and both systems have the same channel realizations.

Now, from the sufficient condition in Theorem~\ref{thm:suff-rf},
policy $\mathbf{P}$ will serve all packets in $B_1(\hat{w}(p))$,
including packet $p$. This contradicts with the hypothesis that 
packet $p$ is not served (by policy $\mathbf{P}$) in the system 
$\bar{Q}_1$ in time-slot $\tau$ (i.e., $p \notin R_1(\tau)$). 

So far, we have shown that for any given sample path and for any 
value of $h$, by the end of any time-slot $t$, policy $\mathbf{P}$ 
has served every packet that the FBS policy has served. 

Next, we want to show that policy $\mathbf{P}$ dominates the version of 
the perfect-matching policy described in Section~\ref{subsec:rfdo}. Note 
that in each time-slot, the packets served by the perfect-matching policy 
are the oldest packets in the system. The difference between FBS and the 
perfect-matching is the following. The HOL frame that can be served by FBS 
has at most $Lh$ packets from each queue and has at most $n_0=n-Lh$ packets 
from the system, while the set of packets that can be served by the 
perfect-matching policy has at most one packet from each queue and has at 
most $n$ packets from the system. Following a similar argument as above for 
the FBS policy, we can show that for any given sample path, by the end of 
any time-slot $t$, policy $\mathbf{P}$ has served every packet that the 
perfect-matching policy has served. This completes the proof.

%%%%%%%%%%%%%%%%%%%%%%%%%%%%%%%%%%%%%%%%%%%%%%%%%%%%%%%%%%%%%%%%
\section{Proof of Proposition~\ref{pro:dwmn}} \label{app:pro:dwmn}
%%%%%%%%%%%%%%%%%%%%%%%%%%%%%%%%%%%%%%%%%%%%%%%%%%%%%%%%%%%%%%%%
We first prove that DWM-$n$ policy is an OPF policy and is thus 
rate-function delay-optimal. The proof follows immediately from a 
property of the MVM in bipartite graphs. We restate this property 
in the following lemma.

\begin{lemma}[Lemma~6 of \cite{spencer84}]
\label{lem:heaviest}
Consider a bipartite graph, and the $k$ heaviest vertices, for some 
$k$. If there is a matching that matches all the heaviest $k$ vertices, 
then any MVM matches all of them too.
\end{lemma}

Since DWM-$n$ policy finds an MVM in the constructed bipartite graph,
Lemma~\ref{lem:heaviest} implies that for any $k \in \{1,2,\dots,n\}$, 
if the $k$ oldest packets can be served by some scheduling policy, then 
DWM-$n$ policy can serve these $k$ packets as well. This completes the 
first part of the proof.

Next, we prove that DWM-$n$ policy has a complexity of $O(n^{2.5} \log n)$.
Note that in order to select the $n$ oldest packets in the system, it is 
sufficient to sort the $n^2$ packets picked by DWM policy, i.e., the $n$ 
oldest packets of each of the $n$ queues, as no other packets can be among 
the $n$ oldest packets in the system. The complexity of sorting $n^2$ packets 
\cite{cormen09} is $O(n^2 \log n)$. Given the $n$ oldest packets in the 
system, DWM-$n$ policy constructs an $n \times n$ bipartite graph and finds 
an MVM \cite{spencer84} in $O(n^{2.5} \log n)$ time. Hence, the overall 
complexity of DWM-$n$ is $O(n^{2.5} \log n)$, which completes the proof.

%%%%%%%%%%%%%%%%%%%%%%%%%%%%%%%%%%%%%%%%%%%%%%%%%%%%%%%%%%%%%%%%%%%%
\section{Proof of Proposition~\ref{pro:unstable}} \label{app:pro:unstable}
%%%%%%%%%%%%%%%%%%%%%%%%%%%%%%%%%%%%%%%%%%%%%%%%%%%%%%%%%%%%%%%%%%%%
The following simple counter-example shows that DWM-$n$ cannot stabilize a 
feasible arrival rate vector, and is thus not throughput-optimal in general.

Consider a system with two queues and two servers, i.e., a system with $n=2$. 
We assume the \emph{i.i.d.} ON-OFF channel model as in Assumption~\ref{ass:channel}, 
i.e., each server is connected to each queue with probability $q \in (0,1)$, 
and is disconnected otherwise. In each time-slot, a server can serve at most 
one packet of a queue that is connected to this server. In such a system, the 
optimal throughput region can be described as $\Lambda^{*} = \{ \lambda ~|~ 
\lambda_1 \le 2q$, $\lambda_2 \le 2q$, and $\lambda_1 +\lambda_2 \le 2(2q - 
q^2) \}$, where the first two inequalities are obvious, and the last inequality
is due to the following. For each of the two servers, the probability that
at least one queue is connected to the server is $2q-q^2$, hence, the service
each server can provide is $2q-q^2$, and the total (effective) capacity is 
thus $2(2q-q^2)$. Note that any arrival rate vector $\lambda$ strictly inside 
the optimal throughput region $\Lambda^{*}$, is feasible.

Next, we construct an arrival process as follows. Consider a frame consisting 
of two time-slots. In each frame, there are packet arrivals to the system with 
probability $p \in (0,1)$, and no arrivals otherwise. In a frame that has arrivals, 
there are $K$ packet arrivals to queue $Q_1$ and no arrivals to queue $Q_2$ in 
the first time-slot, and there are no arrivals to queue $Q_1$ and $K$ packet 
arrivals to queue $Q_2$ in the second time-slot, where we assume that $K \ge 4$. 
This type of arrival process yields an arrival rate vector of $\lambda^*=[\frac 
{pK} {2}, \frac {pK} {2}]$. It is easy to check that $\lambda^*$ is feasible, 
if $pK \le 4q - 2q^2$. 

Now, we characterize an upper bound of the service rate under DWM-$n$ 
policy. Recall that DWM-$n$ considers only the $n$ oldest packets in the 
system and maximizes the sum of the delays of the packets scheduled over
these $n$ packets, and no other packets will be scheduled. Hence, in each 
time-slot, DWM-$n$ considers only the two oldest packets in the system. 
Consider any time-slot $t_1$, where $K-1$ out of the $K$ packets arriving 
to queue $Q_1$ in the same time-slot are still waiting in the system. The 
other one packet could have been scheduled with a packet in $Q_2$, or with 
a packet that arrived to $Q_1$ earlier, or it could have been scheduled 
alone in a time-slot before $t_1$. Note that the first $K-2$ packets out 
of these $K-1$ packets cannot be scheduled with packets in queue $Q_2$, 
due to the operations of DWM-$n$. Hence, in any time-slot $t_2$ before 
these $K-1$ packets are completely evacuated, each server must serve queue 
$Q_1$ if this server is connected to queue $Q_1$, and no server will serve 
$Q_2$ even if this server is connected to queue $Q_2$, as the packets of 
$Q_2$ are not among the two oldest packets in the system in such time-slot 
$t_2$. Hence, the expected service rate for these $K-2$ packets is $2q$, 
and it thus takes $\frac {K-2} {2q}$ time-slots on average to evacuate the 
$K-2$ packets. Similarly, it takes $\frac {K-2} {2q}$ time-slots on average 
to evacuate such $K-2$ packets in queue $Q_2$. Therefore, the total service 
rate of the system under DWM-$n$ is no greater than $\frac{2K}{\frac{2(K-2)}{2q}}=\frac{2qK}{K-2}$. 
It is clear that the system is unstable if the total arrival rate is greater 
than the total service rate, i.e., $pK > \frac {2qK}{K-2}$. Then, by choosing 
$p=\frac{17}{96}$, $q=\frac{1}{2}$ and $K=8$, we obtain a feasible arrival rate 
vector $\lambda^*$ that cannot be stabilized by DWM-$n$. This completes the proof.

%%%%%%%%%%%%%%%%%%%%%%%%%%%%%%%%%%%%%%%%%%%%%%%%%%%%%%%%%%%%%%%%%%%%%%
\section{Proof of Theorem~\ref{thm:suff-to}} \label{app:thm:suff-to}
%\lcsection{PROOF OF THEOREM~\ref{thm:suff-to}} \label{app:thm:suff-to}
%%%%%%%%%%%%%%%%%%%%%%%%%%%%%%%%%%%%%%%%%%%%%%%%%%%%%%%%%%%%%%%%%%%%%%
Suppose that the sufficient condition is satisfied under policy $\mathbf{P}$, 
i.e., there exists a constant $M>0$ such that in any time-slot $t$ and for 
all $j \in \{1,2,\dots,n \}$, queue $Q_{i(j,t)}$ satisfies that $W_{i(j,t)}(t) 
\ge Z_{r,M}(t)$ for all $r \in \Gamma_j(t)$ such that $Q_r(t) \ge M$. We want 
to show that policy $\mathbf{P}$ can stabilize any arrival rate vector $\lambda$ 
strictly inside the optimal throughput region $\Lambda^*$.

Recall that $Q_i(t)$ denotes the queue length of $Q_i$ at the beginning 
of time-slot $t$, $Z_{i,l}(t)$ denotes the delay of the $l$-th packet of 
$Q_i$ at the beginning of time-slot $t$, $W_i(t)$ denotes the HOL delay 
of $Q_i$ at the beginning of time-slot $t$, and $C_{i,j}(t)$ denotes the 
connectivity between queue $Q_i$ and server $S_j$ in time-slot $t$. Let 
$Y_{i,j}(t)$ denote the service of queue $Q_i$ received from server $S_j$ 
in time-slot $t$, i.e., $Y_{i,j}(t)=C_{i,j}(t)$ if server $S_j$ is allocated 
to serve queue $Q_i$, and $Y_{i,j}(t)=0$ otherwise. We define the random 
process describing the behavior of the underlying system as $\Markov = 
(\Markov(t), t=0,1,2,\dots)$, where
\[
\begin{split}
\Markov(t) \triangleq \{ & (Z_{i,1}(t),Z_{i,2}(t),\dots,Z_{i,Q_i(t)}(t)), 
1 \le i \le n ; \\ 
& C_{i,j}(t), 1 \le i \le n, 1 \le j \le n\}. \\
\end{split}
\]
The norm of $\Markov(t)$ is defined as $\| \Markov(t) \| \triangleq 
\sum_{1 \le i \le n}Q_i(t) + \sum_{1 \le i \le n} W_i(t)$. Let 
$\Markov^{(x)}$ denote a process $\Markov$ with an initial condition 
such that
%\begin{equation}
%\label{eq:initconfig} 
\[
\| \Markov^{(x)}(0) \| = x.
\]
%\end{equation}

The following Lemma was derived in \cite{rybko92} for continuous-time 
countable Markov chains, and it follows from more general results in 
\cite{malyshev79} for discrete-time countable Markov chains. 

\begin{lemma}
\label{lem:stab_cri}
Suppose that there exist a real number $\epsilon>0$ and an integer 
$T > 0$ such that for any sequence of processes $\{\Markov^{(x)}(xT), 
x=1,2,\dots\}$, we have
\begin{equation}
\label{eq:stab_cri}
\limsup_{x \rightarrow \infty} \Expect \left[\frac {1} {x} 
\| \Markov^{(x)} (x T) \| \right] \le 1 - \epsilon,
\end{equation}
then the Markov chain $\Markov$ is stable.
\end{lemma}

Lemma~\ref{lem:stab_cri} implies the stability of the network, and a 
stability criteria of type (\ref{eq:stab_cri}) leads to a fluid limit 
approach \cite{dai95} to the stability problem of queueing systems. 

In the following, we construct the fluid limit model of the system as 
in \cite{dai95,andrews04}. We assume that the 
packets present in the system in its initial state $\Markov^{(x)}(0)$ 
arrived in some of the past time-slots $-(x-1), -(x-2), \dots, 0$, 
according to their delays in state $\Markov(0)$. We define another 
process $\System \triangleq \left(A, Q, W, Y \right)$, i.e., a tuple 
that denotes a list of process, and clearly, a sample path of 
$\System^{(x)}$ uniquely defines the sample path of $\Markov^{(x)}$. 
Then, we extend the definition of $\System$ to each continuous time 
$t \ge 0$ as $\System^{(x)}(t) \triangleq \System^{(x)} (\lfloor t 
\rfloor)$, where $\lfloor t \rfloor$ denotes the integer part of $t$. 

Next, we consider a sequence of processes $\{\frac {1} {\xm} \System^{(\xm)}
(\xm \cdot)\}$ that are scaled in both time and space. Then, using the 
techniques of Theorem~4.1 of \cite{dai95} or Lemma~1 of \cite{andrews04}, 
we can show that for almost all sample paths and for any sequence of 
processes $\{\frac {1} {\xm} \System^{(\xm)}(\xm \cdot)\}$, where $\{\xm\}$ 
is a sequence of positive integers with $\xm \rightarrow \infty$, there 
exists a subsequence $\{\xml\}$ with $\xml \rightarrow \infty$ as $l 
\rightarrow \infty$ such that the following convergences hold \emph{uniformly 
over compact (u.o.c.)} interval:
\begin{eqnarray}
&&\frac{1}{\xml}\int_0^{\xml t} A^{(\xml)}_i(\tau)d\tau \rightarrow \lambda_i t, \label{eq:fluid_a}\\
&&\frac{1}{\xml}\int_0^{\xml t} Y^{(\xml)}_{i,j}(\tau)d\tau \rightarrow  \int_0^t y_{i,j}(\tau) d\tau,\\
&&\frac{1}{\xml}Q^{(\xml)}_i(\xml t) \rightarrow  q_i(t). 
\end{eqnarray}
Similarly, the following convergences (which are denoted by ``$\Rightarrow$") 
hold at every continuous point of the limiting function $w_i(t)$:
\begin{equation}
\frac{1}{\xml}W^{(\xml)}_i(\xml t) \Rightarrow  w_i(t). \label{eq:fluid_w}\\
\end{equation}

Any set of limiting functions $(q,y,w)$ is called a \emph{fluid limit}. 
It is easy to show that the limiting functions are Lipschitz continuous 
in $[0,\infty)$, and are thus absolutely continuous. Therefore, these 
limiting functions are differentiable at almost all (scaled) time $t \in 
[0,\infty)$, which we call {\em regular} time. Moreover, the limiting
functions satisfy that
\begin{equation}
\sum_{1 \le i \le n} q_i(0) + \sum_{1 \le i \le n} w_i(0) = 1,
\end{equation}
and that
\begin{equation}
\label{eq:dq}
\frac{d}{dt} q_i(t) = \left\{
\begin{array}{ll}
\lambda_i - \sum_j y_{i,j}(t), & q_i(t) > 0,\\
(\lambda_i - \sum_j y_{i,j}(t))^+, & q_i(t) = 0,
\end{array}
\right. 
\end{equation}
where $(x)^+ \triangleq \max(x,0)$.

We then prove the stability of the fluid limit model using a standard 
Lyapunov technique. We consider a quadratic Lyapunov function in the 
fluid limit model of the system, and show that the Lyapunov function 
has a negative drift when its value is greater than 0, which implies 
that the fluid limit model is stable.

Using a similar argument as in \cite{andrews04,ji13c}, we can show that 
under policy $\mathbf{P}$, there exists a finite time $T_1>0$ such that 
for all $t \ge T_1$, we have 
\begin{equation}
\label{eq:linear}
q_i(t)=\lambda_i w_i(t)
\end{equation}
for all $i$. This linear relation is similar to the Little's law and 
plays a key role in proving stability of the delay-based schemes. We 
omit the proof of this linear relation for brevity and refer readers 
to \cite{andrews04,ji13c}.

Let $V(q(t))$ denote the Lyapunov function defined as
\begin{equation}
\label{eq:lyapunov}
V(q(t)) \triangleq \frac {1} {2} \sum_{i=1}^{n} \frac {q_i^2(t)} {\lambda_i}.
\end{equation}

Suppose that $\lambda$ is strictly inside $\Lambda^{*}$, then there exists 
a vector $\mu \in \Rate$ such that $\lambda_i < \sum_{j=1}^{n} \mu_{i,j}$ 
for all $i$. Let $\beta$ denote the smallest difference between $\lambda_i$
and $\sum_{j=1}^{n} \mu_{i,j}$, i.e., $\beta \triangleq \min_{1 \le i \le n} 
(\sum_{j=1}^{n} \mu_{i,j} - \lambda_i)$. Clearly, we have $\beta>0$. It 
suffices to show that for any $\zeta_1>0$, there exist a $\zeta_2>0$ and 
a finite time $T_2 > 0$ such that for all regular time $t \ge T_2$, 
$V(q(t)) \ge \zeta_1$ implies $\frac{D^+}{dt^+} V(q(t)) \le - \zeta_2$, 
where $\frac{D^+}{dt^+} V(q(t)) = \lim_{\delta \downarrow 0} \frac 
{V(q(t+\delta))-V(q(t))} {\delta}$. Choose any $T_2 \ge T_1$. Since $q(t)$ 
is differentiable for all regular time $t \ge T_2$ such that $V(q(t)) > 0$, 
we can obtain the derivative of $V(q(t))$ as
\begin{equation}
\label{eq:dol}
\begin{split}
\frac{D^+}{dt^+} & V(q(t)) \\
\stackrel{(a)}=& \sum_{i=1}^{n} \frac {q_i(t)} {\lambda_i} \cdot (\lambda_i - \sum_{j=1}^{n} y_{i,j}(t))  \\
\stackrel{(b)}=& \sum_{i=1}^{n} w_i(t) \cdot (\lambda_i - \sum_{j=1}^{n} \mu_{i,j}(t)) \\
& + \sum_{i=1}^{n} w_i(t) \cdot (\sum_{j=1}^{n} \mu_{i,j}(t) - \sum_{j=1}^{n} y_{i,j}(t)), \\
=& \sum_{i=1}^{n} w_i(t) \cdot (\lambda_i - \sum_{j=1}^{n} \mu_{i,j}(t)) \\
& + \sum_{j=1}^{n} \left( \sum_{i=1}^{n} w_i(t) \mu_{i,j}(t) - \sum_{i=1}^{n} w_i(t) y_{i,j}(t) \right), 
\end{split}
\end{equation}
where (a) is from (\ref{eq:dq}), and (b) is from (\ref{eq:linear}) along 
with a little algebra. 

From (\ref{eq:linear}) and (\ref{eq:lyapunov}), we can choose $\zeta_3 > 0$ 
such that $V(q(t)) \ge \zeta_1$ implies $\max_{1 \le i \le n} w_i(t) \ge \zeta_3$.
Then, in the final result of (\ref{eq:dol}), we can conclude that the first 
term is bounded. That is,
\[
\begin{split}
\sum_{i=1}^{n} & w_i(t) \cdot (\lambda_i - \sum_{j=1}^{n} \mu_{i,j}(t)) \\
& \le - \zeta_3 \min_{1 \le i \le n} (\sum_{j=1}^{n} \mu_{i,j}(t) - \lambda_i) \\
& \le - \zeta_3 \beta \\
& \triangleq - \zeta_2 < 0.
\end{split}
\]
Therefore, we have that $\frac{D^+}{dt^+} V(q(t)) \le - \zeta_2$ if the second 
term in the final result of (\ref{eq:dol}) is non-positive. We show this in the 
following. 

Considering the neighborhood around a fixed (scaled) time $t \ge T_2$, we define 
$N \triangleq \{\lceil \xml t \rceil, \lceil \xml t \rceil + 1, \dots, \lfloor 
\xml (t+\delta) \rfloor\}$, where $\delta$ is a small positive number and $\{\xml\}$ 
is a positive subsequence for which the convergence to the fluid limit holds.
We will omit the superscript $(\xml)$ of the random variables (depending on the 
choice of the sequence $\{\xml\}$) throughout the rest of the proof for notational 
convenience (e.g., we use $Q_i(t)$ to denote $Q^{(\xml)}_i(t)$). 
We want to show that under policy $\mathbf{P}$, in each time-slot $\tau \in N$, 
each server $S_j$ serves a connected queue $Q_{i(j,\tau)}$ having the largest 
weight in the fluid limits, i.e., $w_{i(j,\tau)}(t) = L_j(\tau) \triangleq 
\max_{i \in \mS_j(\tau)} w_i(t)$ (recall that $\mS_j(\tau) = \{1 \le i \le n 
~|~ C_{i,j}(\tau)=1 \}$). Note that the trivial statement holds if $\mS_j(\tau) 
= \emptyset$ or $L_j(\tau)=0$. Hence, suppose that $\mS_j(\tau) \neq \emptyset$ 
and $L_j(\tau)>0$. Consider $r,s \in \mS_j(\tau)$ such that $w_s(t)=L_j(\tau)$ 
and $W_r(\tau)=\max_{i \in \mS_j(\tau)} W_i(\tau)$. In other words, $Q_s$ is a 
queue having the largest weight in the fluid limit among all the queues being 
connected to server $S_j$ in time-slot $\tau$, and $Q_r$ is a queue having the 
largest weight in the original discrete-time system among all the queues being 
connected to server $S_j$ in time-slot $\tau$. Note that it is possible that 
$r=s$. Then, for any time-slot $\tau \in N$, we have that 
\begin{equation}
\label{eq:ww}
\begin{split}
& W_{i(j,\tau)} (\lceil \xml t \rceil) \\ 
& \stackrel{(a)}\ge W_{i(j,\tau)}(\tau) - 
(\lfloor \xml (t+\delta) \rfloor - \lceil \xml t \rceil) \\
& \stackrel{(b)}\ge Z_{r,M}(\tau) - (\lfloor \xml (t+\delta) \rfloor - \lceil \xml t \rceil) \\
& \ge Z_{r,M}(\tau) - W_r(\tau) + W_r(\tau) - (\lfloor \xml (t+\delta) \rfloor - \lceil \xml t \rceil) \\
& \stackrel{(c)}\ge Z_{r,M}(\tau) - W_r(\tau) + W_s(\tau) - (\lfloor \xml (t+\delta) \rfloor - \lceil \xml t \rceil) \\
& \stackrel{(d)}\ge Z_{r,M}(\tau) - W_r(\tau) + W_s(\lfloor \xml (t+\delta) \rfloor) \\
&~~~~~~ - 2 (\lfloor \xml (t+\delta) \rfloor - \lceil \xml t \rceil), 
\end{split}
\end{equation}
where (a) and (d) are due to the fact that the HOL delay cannot increase by
more than $\lfloor \xml (t+\delta) \rfloor - \lceil \xml t \rceil$ within $\lfloor 
\xml (t+\delta) \rfloor - \lceil \xml t \rceil$ time-slots, (b) is from the property
of policy $\mathbf{P}$ satisfying the sufficient conditions, and (c) is due to 
$W_r(\tau)=\max_{i \in \mS_j(\tau)} W_i(\tau)$ and $s \in \mS_j(\tau)$. Divide both 
sides of the final result of the above equation by $\xml$ and let $\xml$ goes to 
infinity, we have that 
\begin{equation}
\begin{split}
w_{i(j,\tau)}(t) & \stackrel{(a)}= \lim_{\xml \rightarrow \infty} \frac {W_{i(j,\tau)}(\xml t)} {\xml} \\
& \stackrel{(b)}\ge \lim_{\xml \rightarrow \infty} \frac {Z_{r,M}(\tau) - W_r(\tau)} {\xml} + w_s(t+\delta) - 2 \delta \\
& \stackrel{(c)}= w_s(t+\delta) - 2 \delta, 
\end{split}
\end{equation}
where (a) is from the definition of fluid limits, (b) is from (\ref{eq:ww}) and 
$\lim_{\xml \rightarrow \infty} \frac {\lfloor \xml (t+\delta) \rfloor - \lceil 
\xml t \rceil} {\xml} = \delta$, and (c) is because
$\lim_{\xml \rightarrow \infty} \frac {Z_{r,M}(\tau) - W_r(\tau)} {\xml} = 0$, as 
the SLLN of (\ref{eq:slln}) holds and $Z_{r,M}(\tau) - W_r(\tau)$ is the difference
of the arriving times of two packets having finite number of packets in-between. 
Since the above equation holds for any arbitrarily small positive number $\delta$, 
by letting $\delta$ go to 0 on both sides of the final result of the above equation, 
we have $w_{i(j,\tau)}(t) \ge w_s(t) = L_j(\tau)$, and in particular, we have 
$w_{i(j,\tau)}(t) = L_j(\tau)$. This is true for each $j$ and for each $\tau \in N$. 
Therefore, under policy $\mathbf{P}$, the service vector $y(t)$ satisfies that
\[
\sum_{i=1}^n w_i(t) y_{i,j}(t) = \max_{\nu \in \Rate} \sum_{i=1}^n w_i(t) \nu_{i,j},
\]
for all $j \in \{1,2,\dots,n \}$.

Thus, we have that
\begin{equation}
y(t) \in \argmax_{\nu \in \Rate} \sum_{j=1}^{n} \sum_{i=1}^{n} w_i(t) \nu_{i,j},
\end{equation}
which implies that 
\begin{equation}
\sum_{j=1}^{n} \sum_{i=1}^{n} w_i(t) \mu_{i,j}(t) \le \sum_{j=1}^{n} \sum_{i=1}^{n} w_i(t) y_{i,j}(t).
\end{equation}

Therefore, this shows that $V(q(t)) \ge \zeta_1$ implies $\frac{D^+}{dt^+} 
V(q(t)) \le - \zeta_2$ for all $t \ge T_2$. It immediately follows that for 
any $\zeta>0$, there exists a finite $T \ge T_2 >0$ such that $\sum_{1 \le 
i \le n} q_i(T) \le \zeta$. Further, we have that 
\[
\sum_{1 \le i \le n} (q_i(T) + w_i(T)) \le (1 + \frac{1}{\min_{1 \le i \le n} \lambda_i}) \zeta
\] 
due to the linear relation (\ref{eq:linear}).

Now, consider any fixed sequence of processes $\{\Markov^{(x)}, x=1,2,\dots\}$ 
(for simplicity also denoted by $\{x\}$). From the convergences (\ref{eq:fluid_a})-(\ref{eq:fluid_w}), 
we have that for any subsequence $\{\xm\}$ of $\{x\}$, there exists a further 
(sub)subsequence $\{\xml\}$ such that
\[
\begin{split}
&\lim_{j \rightarrow \infty}  \frac {1} {\xml} \| \Markov^{(\xml)} (\xml T) \| \\
&~~~~~~= \sum_{1 \le i \le n} (q_i(T) +  w_i(T)) \le (1 + \frac{1}{\min_{1 \le i \le n} \lambda_i}) \zeta
\end{split}
\]
almost surely. This in turn implies (for small enough $\zeta$) that
\begin{equation}
\label{eq:mc_conv}
\lim_{x \rightarrow \infty} \frac {1} {x} \| \Markov^{(x)} (x T) \| 
\le (1 + \frac{1}{\min_{1 \le i \le n} \lambda_i}) \zeta \triangleq 
1 - \epsilon < 1
\end{equation}
almost surely. 

We can show that the sequence $\{ \frac {1} {x} \| \Markov^{(x)} (x T) \|, 
x=1,2,\dots \}$ is uniformly integrable, due to the following:
\[
\frac {1} {x} \| \Markov^{(x)} (x T) \| \le 1 + \frac {1} {x} 
\sum_{1 \le i \le n} \int^{xT}_{\tau=0} A_i(\tau) d \tau + nT
\]
and 
\[
\Expect [1 + \frac {1} {x} \sum_{1 \le i \le n} 
\int^{xT}_{\tau=0} A_i(\tau) d \tau + nT] < \infty,
\]
where the above finite expectation is from our assumption on the arrival 
process. Then, the almost surely convergence in (\ref{eq:mc_conv}) along 
with uniform integrability implies the following convergence in the mean:
\[
\limsup_{x \rightarrow \infty} \Expect [ \frac {1} {x} \| 
\Markov^{(x)} (x T) \| ] \le 1- \epsilon.
\]

Since the above convergence holds for any sequence of processes 
$\{\Markov^{(x)}(xT), x=1,2,\dots\}$, the condition of type 
(\ref{eq:stab_cri}) in Lemma~\ref{lem:stab_cri} is satisfied. 
This completes the proof of Theorem~\ref{thm:suff-to}.

%%%%%%%%%%%%%%%%%%%%%%%%%%%%%%%%%%%%%%%%%%%%%%%%%%%%%%%%%%%%%%%%%%%%%%
\section{Proof of Proposition~\ref{pro:dwm-to}} \label{app:pro:dwm-to}
%%%%%%%%%%%%%%%%%%%%%%%%%%%%%%%%%%%%%%%%%%%%%%%%%%%%%%%%%%%%%%%%%%%%%%
We prove it by showing that DWM is an MWF policy.

Let $M=n$. We want to show that the sufficient condition in 
Theorem~\ref{thm:suff-to} is satisfied, i.e., in any time-slot 
$t$ and for all $j \in \{1,2,\dots,n \}$, DWM policy allocates
server $S_j$ to serve queue $Q_{i(j,t)}$, which satisfies that 
$W_{i(j,t)}(t) \ge Z_{r,n}(t)$ for all $r \in \Gamma_j(t)$
such that $Q_r(t) \ge n$.

Suppose that the sufficient condition is not satisfied, i.e., consider any 
server $S_j$ such that $Q_r(t) \ge n$ for some $r \in \Gamma_j(t)$, and $S_j$ 
is allocated to serve queue $Q_{i(j,t)}$, and suppose that $W_{i(j,t)}(t) < 
Z_{r,n}(t)$. Since $Q_r(t) \ge n$ and at most $n-1$ packets could be matched
with the other $n-1$ servers, there must be at least one of the $n$ oldest 
packets in $Q_r$ remaining unmatched. Suppose this packet is the $k$-th oldest 
packet in queue $Q_r$, then $Z_{r,k}(t) \ge Z_{r,n}(t) > W_{i(j,t)}(t)$. Hence, 
DWM policy must match $S_j$ to the $k$-th oldest packet in queue $Q_r$, i.e., 
DWM must allocate $S_j$ to serve $Q_r$ rather than $Q_{i(j,t)}$, which is a 
contradiction.

Therefore, DWM policy is an MWF policy and is thus throughput-optimal.

%%%%%%%%%%%%%%%%%%%%%%%%%%%%%%%%%%%%%%%%%%%%%%%%%%%%%%%%%%%%%%%%%%%%
\section{Proof of Proposition~\ref{pro:dmws-rf}} \label{app:pro:dmws-rf}
%%%%%%%%%%%%%%%%%%%%%%%%%%%%%%%%%%%%%%%%%%%%%%%%%%%%%%%%%%%%%%%%%%%%
By an argument similar to that in Theorem~3 of \cite{bodas10}, we want 
to show that under D-MWS, the delay-violation event occurs with at 
least a constant probability for any fixed delay threshold even if 
$n$ is large.

First, we define $D(x||y) \triangleq x \log \frac {x} {y} + (1-x) \log \frac 
{1-x} {1-y}$. Then, fix any real number $p^{\prime} \in (0,p)$ and any integer 
$T$, and consider any configuration of queues at the end of time-slot $T$.

\noindent {\bf In time-slot $T+1$:} \\
By the Chernoff bound, there exists an integer $N_1$ such that for all
$n \ge N_1$, with probability at least $1-e^{-D(p^{\prime}||p)n}$, at 
least $np^{\prime}$ queues have packet arrivals at the beginning of 
time-slot $T+1$. Define $\mu \triangleq -2/\log(1-q)$ and $\nu 
\triangleq \mu \log n$. Fix an integer $N_2$ such that for all $n 
\ge N_2$, we have $\nu \ge 1$. Sort the queues in the order of priority 
for service under D-MWS, i.e., after sorting, the first queue has the 
largest weight (HOL delay) with the smallest index; the 
second queue has the largest weight with the second smallest index, or 
has the second largest weight with the smallest index if there is only 
one queue having the largest weight; and so on. Let the set of the first 
$\nu$ queues after sorting be $\mathcal{Q}^* \triangleq \{Q_{i_1}, Q_{i_2}, 
\dots, Q_{i_{\nu}}\}$. Let $E_j$ denote the event that server $S_j$ is 
not connected with any of the queues in $\mathcal{Q}^{*}$. Then, $\Prob(E_j)
=(1-q)^{\nu}=(1-q)^{\mu \log n}$, and we have that
\begin{equation}
\Prob \left( \cup_{j=1}^n E_j\right) \le \sum_{j=1}^n \Prob (E_j)
=n(1-q)^{\mu \log n} = \frac {1} {n},
\end{equation}
where the last equality is because $(1-q)^{\mu \log n} = \exp (\mu 
\log n \log (1-q) ) = \frac {1} {n^2}$.
Thus, with probability at least $1-\frac {1} {n}$, each server is connected 
to at least one queue in $\mathcal{Q}^*$. According to the operations of
D-MWS, a server connected to at least one queue in $\mathcal{Q}^*$ must be 
allocated to one of the queues in $\mathcal{Q}^*$. Hence, with probability
at least $1-\frac {1} {n}$, all the servers serve queues in $\mathcal{Q}^*$.
Since $|\mathcal{Q}^*|=\nu$ and with probability at least $1-e^{-D(p^{\prime}||p)n}$, 
at least $np^{\prime}$ queues had packet arrivals, it follows that for $n \ge 
N_3 \triangleq \max \{N_1,N_2\}$, with probability at least $1-\frac {1} {n} 
- e^{-D(p^{\prime}||p)n}$
by the union bound, at the end of time-slot $T+1$ (and at the beginning of 
time-slot $T+2$), the system has at least $np^{\prime} - \nu$ queues having 
a weight (HOL delay) of at least 1. Let this set of queues (of weight 
being 1) be $\mA_1$.

\noindent {\bf In time-slot $T+2$:} \\
By the similar argument above for time-slot $T+1$, it follows that, with probability
at least $1 - \frac {1} {n}$, no more than $\nu$ queues can receive service. 
Combining this with the result for time-slot $T+1$ and using the union bound,
we have that for all $n \ge N_3$, with probability at least $1-\frac 
{2} {n} - e^{-D(p^{\prime}||p)n}$, at the end of time-slot $T+2$ (and at the 
beginning of time-slot $T+3$), there exists a set $\mA_2$ of queues such that 
$|\mA_2| \ge |\mA_1| - \nu \ge np^{\prime} - 2\nu$, and each queue in $\mA_2$ 
has a weight (HOL delay) of at least 2.

Repeating the same argument above, we have that for all $n \ge N_3$, 
with probability at least $1-\frac {b+1} {n} - e^{-D(p^{\prime}||p)n}$, at the 
end of time-slot $T+b+1$, there exists a set $\mA_{b+1}$ of queues such that
$|\mA_{b+1}| \ge np^{\prime} - (b+1)\nu$, and each queue in $\mA_{b+1}$ has a 
weight (HOL delay) of at least $b+1$.

Fix a real number $\epsilon \in (0,1)$, there exists an integer $N_4$ such that
for all $n \ge N_4$, we have $1-\frac {b+1} {n} - e^{-D(p^{\prime}||p)n} \ge \epsilon$
and $np^{\prime} - (b+1)\nu = np^{\prime} - (b+1) \mu \log n \ge 1$. Hence, for
a system with $n \ge N_5 \triangleq \max \{N_3,N_4\}$, starting with time-slot 
$T$, with probability at least $\epsilon$, we have at least one queue having a 
HOL delay of at least $b+1$ at the end of time-slot $T+b+1$ (or at the beginning
of time-slot $T+b+2$). Let $T=-b-2$, then the above result shows that the delay 
violation event occurs with at least a constant probability even if $n$ is large. 
This completes the proof.

%%%%%%%%%%%%%%%%%%%%%%%%%%%%%%%%%%%%%%%%%%%%%%%%%%%%%%%%%%%%%%%%%%
\section{Proof of Theorem~\ref{thm:hybrid}} \label{app:thm:hybrid}
%%%%%%%%%%%%%%%%%%%%%%%%%%%%%%%%%%%%%%%%%%%%%%%%%%%%%%%%%%%%%%%%%%
We first show that a hybrid OPF-MWF policy is an (overall) OPF policy 
and is thus rate-function delay-optimal. Note that in stage~1, the 
operations of an OPF policy already guarantees that the sufficient 
condition in Theorem~\ref{thm:suff-rf} is satisfied. Since in stage~2, 
the matched servers and packets in stage~1 will not be considered, 
it ensures that the operations do not perturb the satisfaction of the 
sufficient condition for rate-function delay optimality. 

In the following, we want to show that a hybrid OPF-MWF policy is an 
(overall) MWF policy and is thus throughput-optimal. Let $M=n$. We want 
to show that the sufficient condition in Theorem~\ref{thm:suff-to} is 
satisfied, i.e., in any time-slot $t$ and for all $j \in \{1,2,\dots,n \}$, 
a hybrid OPF-MWF policy allocates server $S_j$ to serve queue $Q_{i(j,t)}$, 
which satisfies that $W_{i(j,t)}(t) \ge Z_{r,n}(t)$ for all $r \in 
\Gamma_j(t)$ such that $Q_r(t) \ge n$.

First, we want to show that in stage~1, an OPF policy also guarantees that
\emph{all allocated servers in stage~1 satisfies the sufficient condition 
for throughput optimality.} Consider each server $S_l$ such that $l \in 
\{1,2,\dots,n \} \backslash R(t)$, i.e., all servers $S_j$ that are 
allocated in stage~1. Then, $Q_{i(l,t)}$ is the queue served by server $S_l$ 
in stage~1 of time-slot $t$. Since we run an OPF policy in stage~1, server 
$S_l$ serves a packet among the $n$ oldest packets in the system, and it 
must satisfy that $W_{i(j,t)}(t) \ge Z_{r,n}(t)$ for any $r \in \Gamma_l(t)$ 
such that $Q_r(t) \ge n$.

Next, consider each server $S_j$ such that $j \in R(t)$, then $Q_{i(j,t)}$ 
is the queue served by server $S_j$ in stage 2 of time-slot $t$. It is 
clear from Condition 2) of Definition~\ref{def:hybrid} that $W_{i(j,t)}(t) 
\ge Z_{r,n}(t)$ for all $r \in \Gamma_j(t)$ such that $Q_r(t) \ge n$.

Therefore, a hybrid OPF-MWF policy is an (overall) MWF policy and is 
thus throughput-optimal.

%%%%%%%%%%%%%%%%%%%%%%%%%%%%%%%%%%%%%%%%%%%%%%%%%%%%%%%%%%%%%%%%%%%%
\section{Proof of Theorem~\ref{thm:dwm-mws}} \label{app:thm:dwm-mws}
%%%%%%%%%%%%%%%%%%%%%%%%%%%%%%%%%%%%%%%%%%%%%%%%%%%%%%%%%%%%%%%%%%%%
To show that DWM-$n$-MWS is a hybrid OPF-MWF policy, it is sufficient to show 
that Condition 2) of Definition~\ref{def:hybrid} is satisfied.

Given any time-slot $t$, consider each server $S_j$ such that $j \in R(t)$, then 
$Q_{i(j,t)}$ is the queue served by server $S_j$ in stage~2 under D-MWS. Let $M=n$. 
We want to show that $W_{i(j,t)}(t) \ge Z_{r,n}(t)$ for all $r \in \Gamma_j(t)$ 
such that $Q_r(t) \ge n$. 

Let $W^{\prime}_i(t)$ be the HOL delay of queue $Q_i$ at the beginning of 
stage~2. Let $\Gamma^{\prime}_j(t)$ denote the set of queues that are 
connected to server $S_j$ and have the largest weight among the connected queues 
at the beginning of stage 2 of time-slot $t$, i.e., $\Gamma^{\prime}_j(t) 
\triangleq \{i \in \mS_j(t) ~|~ W_i^{\prime}(t) = \max_{l \in \mS_j(t)} W_l^{\prime}(t) 
\}$, where $\mS_j(t)= \{1 \le i \le n ~|~ C_{i,j}(t)=1 \}$. According to the operations
of D-MWS, the index of queue that is served by server $S_j$ satisfies that $i(j,t) = 
\min \{i ~|~ i \in \Gamma^{\prime}_j(t) \}$, hence, we have $W^{\prime}_{i(j,t)}(t) 
= W^{\prime}_r(t)$ for any $r \in \Gamma^{\prime}_j(t)$. This implies that 
$W_{i(j,t)}(t) \ge W^{\prime}_{i(j,t)}(t) = W^{\prime}_r(t) \ge Z_{r,n}(t)$ 
for any $r \in \Gamma^{\prime}_j(t)$ such that $Q_r(t) \ge n$, where the last 
inequality is because $Q_r(t) \ge n$ and thus the HOL packet of queue $Q_r$ at 
the beginning of stage~2 must not have a later position than the $n$-th 
packet in queue $Q_r$ at the beginning of time-slot $t$. This holds for all $j 
\in R(t)$ and any time-slot $t$. Therefore, DWM-$n$-MWS is a hybrid OPF-MWF policy.

Since the complexity of DWM-$n$ and D-MWS is $O(n^{2.5} \log n)$ and $O(n^2)$, 
respectively, the overall complexity of DWM-$n$-MWS policy is $O(n^{2.5} \log n)$. 

% \section*{Acknowledgment}
% The authors would like to thank the associate editor and the reviewers 
% for their valuable comments. 

%\bibliographystyle{abbrv}
\bibliographystyle{IEEEtran}
\bibliography{ofdm}

\end{document}